\documentclass[twoside]{article}
\usepackage{amsmath, amsthm, amsfonts,amssymb, dsfont, stmaryrd}
\usepackage{hyperref}
\usepackage{xcolor}
\usepackage{graphicx}
\usepackage[ruled,vlined]{algorithm2e} 
\usepackage{enumitem}
\def\dblbr#1{\llbracket #1\rrbracket}

\usepackage[accepted]{aistats2026}
\usepackage[round]{natbib}

\graphicspath{{figs/}}

\newtheorem{theorem}{Theorem}[section]   
\newtheorem{proposition}[theorem]{Proposition}
\newtheorem{lemma}[theorem]{Lemma}
\newtheorem{corollary}[theorem]{Corollary}
\newtheorem{condition}[theorem]{Condition}

\newtheorem{definition}[theorem]{Definition}

\begin{document}

%
\runningtitle{Joint Graph Sorting for Graphon Estimation}

%

\twocolumn[

\aistatstitle{Low-Complexity and Consistent Graphon Estimation\\ from Multiple Networks}

\aistatsauthor{Roland B. Sogan \And Tabea Rebafka}

\aistatsaddress{LPSM, Sorbonne Universit\'e, France  \And MIA-PS, AgroParisTech, France}]

\begin{abstract}
Recovering the random graph model from an observed collection of networks is known to present significant  challenges in the setting, where  the  networks do not share a common node set and have different sizes. More specifically, the goal is the estimation of the  graphon function that parametrizes the nonparametric exchangeable random graph model. 
Existing methods typically suffer from either limited accuracy or high computational complexity. We introduce a new histogram-based estimator with low algorithmic complexity that achieves high accuracy by jointly aligning the nodes of all graphs, in contrast to most conventional methods that order nodes graph by graph. Consistency results of the proposed graphon estimator are established. 
A numerical study shows that the proposed estimator outperforms  existing methods in terms of accuracy, especially when the dataset comprises only small and  variable-size networks. Moreover, the computing time of the new method is considerably shorter than that of other consistent methodologies. Additionally, when applied to a graph neural network classification task, the proposed estimator enables more effective data augmentation, yielding improved performance across diverse real-world datasets.
\end{abstract}

\section{INTRODUCTION}  
 
Machine learning research on network-structured data has traditionally focused on the analysis of a single graph \citep{Goldenberg2010,Kolaczyk2009}. Increasingly though, modern applications demand methods that can handle collections of networks sharing structural properties, such as in neuroscience, ecology, sociology, and bioinformatics \citep{Fornito2013,Weber19,Poisot2016}. It is therefore more important than ever to develop new statistical approaches that are computationally tractable and capable of jointly leveraging multiple networks to recover the common generative model. Such tools promise to improve performance in any setting where graphs arise under similar   structural assumptions.

Graphons provide a powerful nonparametric model for representing complex network structures. Formally, a graphon is a bivariate symmetric measurable function, interpretable both as the limit of dense graph sequences \citep{Lovasz2006} and as the generative mechanism of exchangeable random graphs \citep{Diaconis2007,Bickel2009}. They have recently garnered increased attention in statistics and machine learning \citep{Eldridge2016,Ruiz2020,Ruiz2023}, being used for inferring network topology, comparing graphs, clustering, link prediction, and beyond. Estimating a graphon is challenging in itself and its inherent non-identifiability adds another layer of difficulty. Indeed, distinct graphons can induce the same distribution of observed networks, so one can only recover an equivalence class of functions rather than a unique truth \citep{Borgs2015,Sischka2022}. While these issues are challenging for a single network,  they are significantly amplified when multiple networks with disjoint node sets of variable size     must be analyzed jointly, where alignment and aggregation of information become  delicate. Thus, joint graphon estimation from multiple networks is particularly demanding.

There is a long history of graphon inference methods in the literature, but most of them focus exclusively on the single-network setting \citep{Keshavan2010,Chatterjee2015,Airoldi2013,Chan2014,Yang2014,Wolfe2013}. When it comes to multiple networks, a straightforward strategy is to estimate a graphon separately on each graph and then average these individual estimates. However, this approach has major drawbacks: 
first, we may encounter difficulties related to the lack of  identifiability; 
second, graphon estimators computed on a small network are very biased and have a strong impact on the mean graphon estimator; third,  the heterogeneity in network sizes is overlooked, since  all graphs are treated equally regardless of scale. 
Advances such as Smoothed Gromov-Wasserstein Barycenters (GWB) \citep{Xu2021} align graphs via Gromov–Wasserstein by representing graphons as step functions on a fixed grid, but remain computationally intensive for large graph collections and enforce a rigid resolution. Neural network–based approaches \citep{Xia2023,Azizpour2025} introduce more flexibility, yet their training complexity makes them impractical for collections of large or numerous networks. Consequently, the problem of joint graphon estimation from a collection of graphs remains a significant open challenge, particularly demanding methods that are both statistically sound and computationally efficient.
\paragraph{Main contributions }
Building on these observations, we introduce the \textbf{Joint Graph Sorting} (JGS) estimator, a novel framework for graphon inference from collections of unaligned graphs with heterogeneous sizes. JGS extends the classical histogram estimator \citep{Chan2014} while maintaining computational efficiency, and overcomes key limitations of existing multi-network approaches. Instead of estimating a graphon separately for each graph, JGS performs a joint alignment of all nodes across networks to construct a single global histogram estimate. This unified approach leads to a more efficient use of information. We establish the consistency of the estimator under mild regularity conditions and validate its practical advantages through extensive numerical experiments. Our main contributions are:
	\begin{enumerate}
		\item We introduce a deterministic, non-iterative joint sorting algorithm that simultaneously infers latent positions and produces a unified graphon estimate from the entire collection of networks.
		\item We establish finite-sample guarantees for the latent position estimates and prove consistency for the associated graphon estimator.
		\item Through extensive numerical experiments, we demonstrate that JGS consistently outperforms state-of-the-art baselines, particularly for collections of small graphs with heterogeneous sizes. The implementation of JGS and the scripts used to reproduce the experiments are publicly available at \url{https://github.com/RolandBonifaceSogan/JGS-graphon}.
		\item We show that JGS executes significantly faster than other consistent methods, making it suitable for large-scale multi-network applications.
	\end{enumerate}
\paragraph{Related Works }\label{related_works}

	Classical graphon estimators such as Stochastic Blockmodel Approximation (SBA) \citep{Airoldi2013}, Sorting and Smoothing (SAS) \citep{Chan2014}, and Universal Singular Value Thresholding (USVT) \citep{Chatterjee2015} are well established for single networks. The SBA method approximates the graphon by a block (step-function) model, partitioning nodes into communities and estimating connection probabilities between blocks. SAS first sorts nodes by empirical degree and then fits a histogram smoothed by total variation minimization. USVT uses singular value decomposition and applies a universal threshold to reconstruct a low-rank approximation of the adjacency matrix. While effective in the single-graph case, these methods are not designed to jointly align or leverage information across multiple graphs that lack node correspondence.
	
	For collections of unaligned networks, several approaches have been developed. \cite{Xu2021} propose Smoothed Gromov-Wasserstein Barycenters (GWB), which aligns histogram-approximated graphons via the Gromov–Wasserstein distance and obtains the estimate as their barycenter. This strategy accommodates heterogeneous network sizes but is computationally intensive and enforces a common discrete resolution across all graphs. Neural network-based methods offer greater flexibility: \cite{Xia2023} introduce Implicit Graphon Neural Representation (IGNR), an implicit neural representation that models the graphon continuously and leverages Gromov–Wasserstein alignment. \cite{Azizpour2025} propose Scalable Implicit Graphon Learning (SIGL), combining implicit neural representations \citep{Sitzmann2020} with graph neural networks \citep{Wu2021} to parametrize the graphon and learn node ordering.
	
	Other approaches model network collections using stochastic block models (SBMs) \citep{Snijders2001,Celisse2012,Abbe2018}. \cite{Chabertliddell2024} extend SBMs to multiple networks by enforcing a shared connectivity structure estimated via variational EM with penalized likelihood for model selection. \cite{Rebafka2024} proposes hierarchical agglomerative clustering based on SBM mixtures, merging networks progressively under integrated classification likelihood while addressing label-switching. Although these methods aim to uncover global latent structures, their computational cost scales poorly with large network collections or large graph sizes.

\section{BACKGROUND}
\paragraph{Graphon } A \textit{graphon} is a symmetric measurable function
$\ensuremath{W:[0,1]^2\to[0,1]}$ encoding edge probabilities.
To generate a   random graph $G=(V,E)$ with $n$ nodes from graphon $W$, first  latent positions are drawn independently as
\begin{equation}
		U_1,\ldots,U_n \;\stackrel{\text{i.i.d.}}{\sim}\; \mathrm{Uniform}[0,1].
\end{equation}
Then, conditionally on $U_{1:n}=(U_1,\ldots,U_n)$,   the adjacency matrix $A$ is generated as

\begin{align}
	A_{i,j}\mid U_{1:n} &\sim \operatorname{Bernoulli}\!\big(W(U_i,U_j)\big)
	&&\text{for } i<j,\nonumber\\
	 A_{j,i}&=A_{i,j} &&\text{for } i<j,\nonumber\\
		A_{i,i}&=0 &&\text{for all } i.
		\label{eq:graphon_model}
\end{align}

We denote by $\mathbb{P}_W$ the distribution of the associated random graph. The model is  called   the \emph{exchangeable random graph model}, a flexible nonparametric framework that subsumes many familiar network models. In particular, the Erd\H{o}s--R\'enyi model corresponds to a constant graphon $W(u,v) = p$, while stochastic block models arise from piecewise-constant graphons  \citep{Matias2014}.
A key difficulty is that the mapping $W \mapsto \mathbb{P}_W$ is not injective: 
different graphons may generate the same distribution of random graphs. 
This non-identifiability stems from measure-preserving relabelings of the latent space, as explained below.  

\paragraph{Graphon identifiability}
Formally, for any measure–preserving map $\varphi:[0,1]\to[0,1]$, the rearranged graphon 
$W^{\varphi}(u,v)=W(\varphi(u),\varphi(v))$ induces the same distribution of exchangeable 
graphs as $W$ \citep{Lovasz2006,Diaconis2007}. Thus, graphons are identifiable only up to 
such transformations $\varphi$.

This non-identifiability is closely related to the fact that graph nodes do not possess a
natural ordering. Indeed, permuting the labels of the nodes of a graph simply corresponds
to permuting the rows and columns of its adjacency matrix, without changing the underlying
network structure. In the graphon framework, measure–preserving transformations play the
continuous analogue of such permutations: applying $\varphi$ rearranges the latent
positions $U_i$ while leaving the distribution of the generated graphs unchanged.

This phenomenon is similar to the label-switching problem in classical finite mixture
models, where there is no natural ordering of the mixture components.


In the literature, a common condition used to obtain an identifiable
representation of a graphon is based on the \emph{normalized degree
function}, defined by
\begin{equation}\label{eq:degree-func}
	g(u) = \int_0^1 W(u,v)\,dv, \qquad u\in[0,1].
\end{equation}
Note that $n g(u)$ corresponds to the expected degree of a node with latent
position $u$ in a graph with $n$ nodes. A sufficient condition ensuring
identifiability of the graphon is the strict monotonicity of the normalized
degree function \citep{Bickel2009}. This condition selects a canonical
representative within the equivalence class of graphons that generate the
same distribution of exchangeable graphs.

\begin{condition}[Strict monotonicity of degrees]\label{strict}
There exists a measure–preserving transformation $\varphi^*$ such that the
normalized degree function $g^*$ associated with $W^{\varphi^*}$ is strictly
increasing on $[0,1]$.
\end{condition}

Under this condition, the transformation $\varphi^*$ is unique almost
everywhere, which implies that the corresponding graphon $W^{\varphi^*}$ is
uniquely defined. We emphasize that this assumption is restrictive: as
shown for instance in \citet{Janson2020}, not every graphon admits a
measure–preserving transformation for which the degree function becomes
strictly increasing. Nevertheless, this condition is widely adopted in the
literature since it provides a convenient canonical representation that
enables ordering the latent positions of nodes. In our framework, this
ordering plays a key role for jointly aligning nodes across networks.

\section{JOINT GRAPH SORTING ESTIMATOR}

\begin{figure*}[!htbp]
	\centering
	\includegraphics[width=\textwidth]{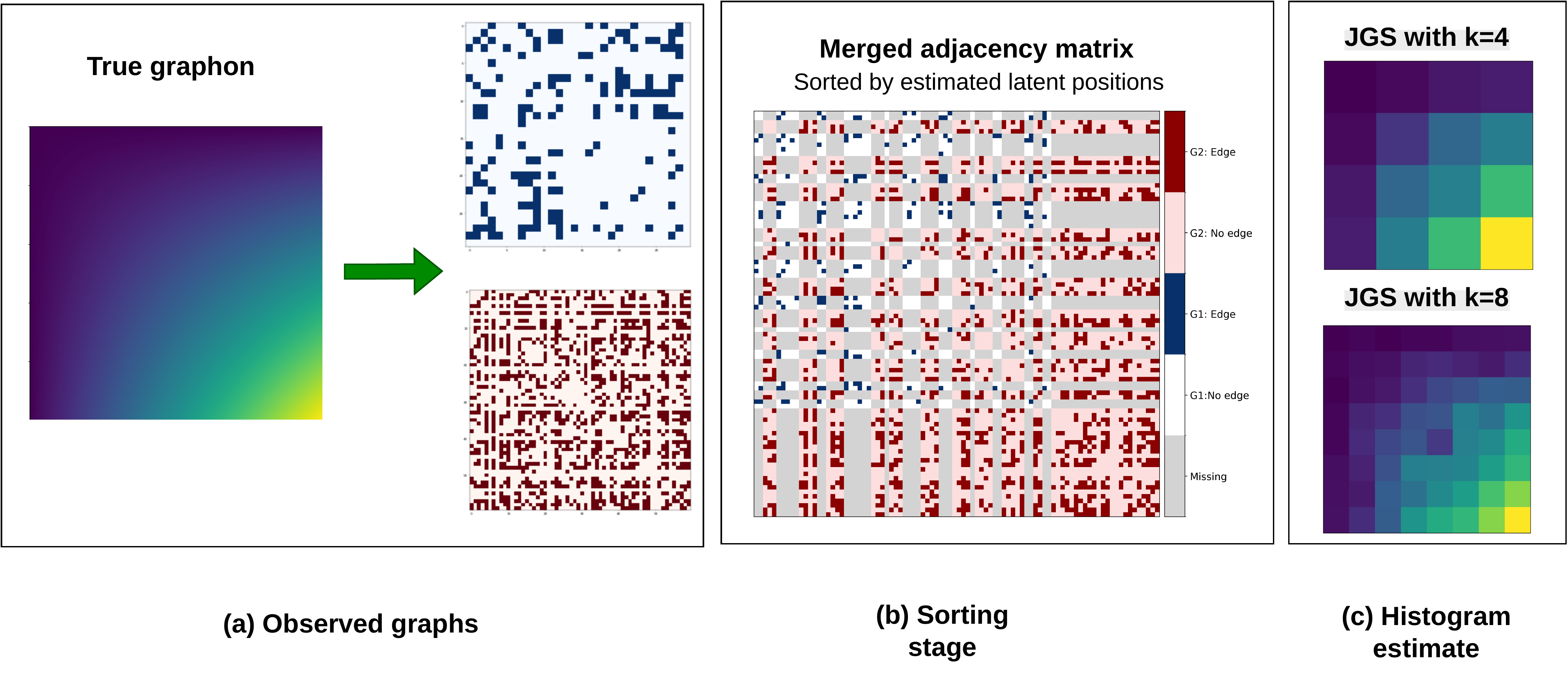}
	\caption{Illustration of the Joint Graph Sorting (JGS) estimator.
		(a) Networks are generated from a graphon.
		(b) Nodes are jointly sorted across networks using normalized degrees,
		yielding a merged adjacency matrix with missing entries (grey).
		Dark blue and white denote edges and non-edges in the first graph, respectively,
		while dark red and misty rose denote edges and non-edges in the second graph, respectively.
		(c) JGS graphon estimates with different resolutions $k$, where colors represent estimated edge probabilities.}
	\label{fig:JGS}
\end{figure*}

We observe a collection of $M$    networks
\[
\mathcal{G}=(G^{(1)},\dots,G^{(M)}),
\]
where the $m$-th graph $G^{(m)}=(V^{(m)},E^{(m)})$ has
node set
$V^{(m)}$ of size $n^{(m)}=\lvert V^{(m)} \rvert$, edge set $E^{(m)}$ and 
adjacency matrix
$A^{(m)} \in \{0,1\}^{n^{(m)} \times n^{(m)}}$.  
Every network $G^{(m)}$ is sampled independently from the same underlying graphon $W$, with latent positions
$U_{i,(m)}$ for $i=1,\dots,n^{(m)}$. In other words, node sets $V^{(m)}$ are distinct and may have different sizes. Our goal is to estimate~$W$ from the  observed collection
$\mathcal{G}$ by leveraging information across all networks.
This section  addresses this problem   by introducing the new
\emph{Joint Graph Sorting (JGS) estimator} and establishing its asymptotic properties.

\subsection{Estimation}

A natural baseline in the multiple network setting is to apply, to each network,
the SAS histogram estimator introduced by \cite{Chan2014}, and then average
the resulting estimates, as implemented for instance in \cite{Xu2021}. A brief description of
these single-network graphon estimation methods is provided in the
''Related Works'' part bellow. However, this
approach overlooks the fact that information can be pooled more efficiently  across
networks by aligning all nodes simultaneously, rather than one graph at a time.  

We   introduce the \emph{Joint Graph Sorting (JGS)} procedure, which consists in first ordering all nodes in the collection according to their normalized empirical degree and then constructing a single histogram estimator of the underlying graphon. The global sorting improves the estimation of the latent positions,   yielding a more accurate graphon estimate.  

Formally,   \emph{normalized empirical degrees} are defined by
\begin{equation}
    \widehat d_{i,(m)}^{\mathrm{norm}}=\frac{1}{n^{(m)}-1}\sum_{j=1}^{n^{(m)}}A_{i,j}^{(m)}, \quad  i\in\dblbr {n^{(m)}}, m\in\dblbr M.
\end{equation}
Nodes of the union of all node sets $\cup_{m\in\dblbr M}V^{(m)}$  are ordered by their normalized empirical degree $d_{i,(m)}^{\mathrm{norm}}$ in increasing order.
Let  $\widehat \sigma^{\mathrm{JGS}}$ be the permutation associated with this sorting such that   $(\widehat \sigma^{\mathrm{JGS}})^{-1}(i, m)$ provides, for node $i$ of the $m$-th network, the node's rank among the $N=\sum_{m\in\dblbr M}n^{(m)}$ nodes in the collection. 
Then a natural estimate of the latent position $U_{i,(m)}$  is given by
\begin{equation}
    \widehat U_{i,(m)}^{\mathrm{JGS}}= \frac{(\widehat \sigma^{\mathrm{JGS}})^{-1}(i,m)}N -\frac1{2N}.
\end{equation}
Note that the estimates $\widehat U_{i,(m)}^{\mathrm{JGS}}$ are equally spaced on the interval $[0,1]$ having a finer resolution than estimates computed on a single network. Specifically, the spacing is $1/N$ instead of $1/n^{(m)}$, which makes it possible to obtain more accurate estimates.  
Now let $A^{\hat\sigma^\mathrm{JGS}}$ be the   $N\times N$ adjacency
matrix  representing the union of all observed networks with sorted nodes. Notably, many entries of   $A^{\widehat \sigma^{\mathrm{JGS}}}$ are
\emph{missing}, because edges exist only among nodes of the same  original network. 

To estimate   graphon $W$,
the  square $[0,1]^2$ is partitioned   into $k^2$ squares, each with a side-length $1/k$ and the observed edge frequencies  on these squares are  computed using the large sorted adjacency matrix  $A^{\hat\sigma^\mathrm{JGS}}$
by
\begin{equation}\label{eqglobal}
\widehat{H}^{\mathrm{JGS}}_{s,t}=
\dfrac{\sum_{i=a_s}^{b_s} \sum_{j=a_t}^{b_t}
A^{\widehat \sigma^\mathrm{JGS}}_{i,j}
\mathds{1}\!\left\{ A^{\widehat \sigma^\mathrm{JGS}}_{i,j}\in\{0,1\}\right\}}
{\max\!\left(1,\sum_{i=a_s}^{b_s} \sum_{j=a_t}^{b_t}
\mathds{1}\!\left\{ A^{\widehat\sigma^\mathrm{JGS}}_{i,j}\in\{0,1\}\right\}\right)}.
\end{equation}
with $a_s = \lfloor N(s-1)/k \rfloor+1$, $b_s = \lfloor Ns/k \rfloor$, for $s,t\in\dblbr k$.
The indicator functions in  \eqref{eqglobal} help dealing with   missing values in $A^{\hat\sigma^\mathrm{JGS}}$.
The final  graphon estimate of $W$ is given by
\begin{equation}
    \widehat{W}^{\mathrm{JGS}}(u,v)=\widehat{H}_{s,t}^{\mathrm{JGS}}, \qquad \text{ if } (u,v) \in I_s \times I_t.
\end{equation}
Figure~\ref{fig:JGS}\,(b) illustrates this histogram estimator: after
joint sorting, the merged matrix $A^{\widehat \sigma^{\mathrm{JGS}}}$  has  large grey regions
corresponding to missing  inter–graph edges; only the colored within–graph fields contribute to the edge frequencies by block $\widehat{H}^{\mathrm{JGS}}_{s,t}$.

\paragraph{Optional smoothing step.}  
To obtain a more regular graphon estimator and satisfy possible smoothness assumptions on the underlying graphon, an optional smoothing step may be applied to the histogram estimate $\widehat{W}^{\mathrm{JGS}}$. This   reduces artefacts and discontinuities while preserving important global structures. Among the possible methods are those based on total variation minimization (e.g.\ as in SAS \citep{Chan2014}) or variational image-processing techniques \citep{Chambolle2004}. The choice of the smoothing method depends on the desired properties such as fidelity to the data, edge preservation, and the trade-off between bias and variance.

\paragraph{Selection of the number of blocks \(k\).}
Choosing the number of  blocks \(k\) is delicate yet crucial for the performance of the graphon estimator.  Based on heuristic arguments,  we propose a  selection rule for the choice of $k$. 
First, to balance bias and variance, when the graphon $W$ is continuous, the optimal choice of $k$ is  expected to verify 
\begin{equation}
    k  \;\asymp\; S^{\frac14},
\end{equation}
where \(S = \sum_{m=1}^M (n^{(m)})^2\) is the number of observed edges and non edges in the graph collection $\mathcal G$ (without counting the missing values).  Second, in the multiple-network setting
one must guard against empty histogram blocks; a sufficient condition to avoid empty blocks is
\begin{equation}
    \frac{N}{k} \;\ge\; c\,(M + \log N),
\end{equation}
where \(c\) is a positive constant.  Combining these arguments leads to the practical choice 
\begin{equation}\label{eq:kopt}
    k \;\asymp\; \min\!\left\{\, S^{1/4},\;\; \frac{N}{c\,(M+\log N)} \,\right\}.
\end{equation}
In experiments, \(c=2\) is found to be a convenient value.   More  details are provided  in the supplementary material. 
This selection rule for  \(k\)  makes the graphon estimator  $\widehat{W}^{\mathrm{JGS}}$   highly accurate.

\paragraph{Complexity. } 
First, concerning the efficient implementation of the graphon estimator $\widehat{W}^{\mathrm{JGS}}$, note that for memory reasons it is not advisable to create the large matrix $A^{\hat\sigma^\mathrm{JGS}}$. Instead, the computation of the edge frequencies
$\widehat{H}_{s,t}^{\mathrm{JGS}}$ is implemented more efficiently as
\begin{equation*} 
\scalebox{0.9}{$\widehat{H}^{\mathrm{JGS}}_{s,t}=
\dfrac{\sum_{m\in\dblbr M} \sum_{i \in \widehat J_s^{(m)}} \sum_{j \in \widehat J_t^{(m)}} A^{(m)}_{i,j} }{\max\!\left(1,\sum_{m\in\dblbr M} |\widehat J_s^{(m)}||\widehat J_t^{(m)}| \right)},\quad s,t\in\dblbr k,$}
\end{equation*}	
where $\widehat J_s^{(m)}$ is the set of node indices $i$ of network $m$ such that $\widehat U_{i,(m)}^{\mathrm{JGS}}\in I_s$.

We analyze the complexity of JGS step by step. The computation of the normalized empirical degrees requires accessing all observed entries of the adjacency matrices and therefore has complexity $\mathcal{O}(S)$. Note that each observed edge contributes to the degrees of its two incident vertices. 

The estimation of latent positions requires $\mathcal{O}(N\log N)$ operations due to the sorting step, plus an additional $\mathcal{O}(N)$ operations to assign the estimated positions. 

For the histogram estimation step, the algorithm processes each of the $S$ observed adjacency entries only once. For each entry, it uses the estimated latent positions to determine its corresponding block $(s,t)$ in constant time and then updates the relevant sums in the numerator and denominator of $\widehat{H}^{\mathrm{JGS}}_{s,t}$. Therefore, the complexity of building the histogram is $\mathcal{O}(S)$. 

As a result, the overall computational complexity of JGS is
\(
\mathcal{O}\!\left(N\log N + S\right).
\).


\subsection{Theoretical results}

This section  studies the asymptotic properties of the JGS graphon estimator $\widehat W^\mathrm{JGS}$. The consistency of the estimator relies on the accurate recovery of the latent positions \(\{U_i^{(m)}, i\in\dblbr{n^{(m)}}, m\in\dblbr M\}\) for all nodes in the network collection. The following proposition states that, under suitable regularity assumptions on the normalized  degree function \(g\), the latent position estimators $\widehat U_{i,(m)}^\mathrm{JGS}$   are close to their true values.

\begin{proposition}[Consistency of latent positions' estimators]
	\label{prop:latent_consistency}
	Let the normalized degree function
\(
g(u)=\int_0^1 W(u,v)\,dv
\)
be strictly increasing. Assume that there exists a constant \(L_1>0\) such that
\[
L_1|x-y|\le |g(x)-g(y)|,\qquad \forall x,y\in[0,1].
\]
	Then, with probability at least \(1 - \delta\), for all 
	\( m \in \dblbr{M} \) and for all \( i \in \dblbr{n^{(m)}} \), the following holds
\begin{align*}
	\left| \widehat{U}_{i,(m)}^{\mathrm{JGS}} - U_{i,(m)} \right| 
	&\le \eta_i^{(m)},
\end{align*}
	where
	\begin{align}\label{eq:eq_eta_im}
	&\eta_i^{(m)}= \frac{2}{L_1} \left( \varepsilon_i^{(m)} + \frac{1}{N} \sum_{(j,\ell) \ne (i,m)} \varepsilon_j^{(\ell)} \right) 
	+ \varepsilon^{(N)},  
	\end{align} 
	\(\varepsilon_j^{(\ell)}= 
	\sqrt{\frac{\log(2/\delta)}{2(\,n^{(\ell)}-1)}}\) and $\varepsilon^{(N)}=\sqrt{ \frac{ \log(2/\delta) }{2N}}$.
\end{proposition}

The key mechanism behind Proposition~\ref{prop:latent_consistency} is the
concentration of the empirical normalized degrees 
$\widehat d_{i,(m)}^{\mathrm{norm}}$, which are unbiased estimators of the theoretical degrees $g(U_{i,(m)})$. 
The strict monotonicity of $g$, together with the lower Lipschitz condition, then allows us to translate degree convergence into latent position convergence.
This explains why the consistency of the JGS latent position estimators $\widehat U_{i,(m)}^\mathrm{JGS}$ is entirely driven by the concentration of the degrees.

A direct consequence of the proposition is the uniform consistency of the latent position estimators in different asymptotic regimes. We consider monotone sequences of network collections, that is, either the networks grow by adding new nodes (i.e., $n^{(m)}$ increases for at least one $m$), or the collection grows by adding new graphs (i.e., $M$ increases), or both.

\begin{corollary}[Uniform consistency of the latent positions' estimators]
	\label{cor:latent_consistency}
	Under the assumptions of Proposition~\ref{prop:latent_consistency},  
	\begin{enumerate}
	\item[(i)] \textbf{A graph size diverges.} 
		If there is a network $m^\star$ such that $n^{(m^\star)} \to \infty$, the estimator of the latent positions of the nodes in the $m^\star$-th graph converges uniformly to the true positions. That is,
		\[
		\max_{i \in \dblbr{n^{(m^\star)}}}
		\left| \widehat U_{i,(m^\star)}^{\mathrm{JGS}} - U_{i,(m^\star)} \right|
		\;\xrightarrow{\;\mathbb{P}\;}\;0.
		\]
    \item[(ii)] \textbf{All graph sizes diverge.} 
		Let the number of graphs $M$ be fixed. If $n^{(m)} \to \infty$ for every $m \in \dblbr{M}$, the latent positions are estimated uniformly consistently
		over the entire collection of networks. That is
		\[
		\max_{m \in \dblbr{M}} \; \max_{i \in \dblbr{n^{(m)}}} 
		\left| \widehat U_{i,(m)}^{\mathrm{JGS}} - U_{i,(m)} \right|
		\;\xrightarrow{\;\mathbb{P}\;}\;0,
		\]
	\end{enumerate}
\end{corollary}

Notably, these results do not cover the case where  the number of graphs $M$ tends to infinity, while the graph sizes $n^{(m)}$ are bounded. Indeed, our simulations suggest that in this case the estimator is not consistent, basically because the normalized degree estimators $\widehat d_{i,(m)}^{\mathrm{norm}}$ are no longer consistent. 

In summary,  the consistency of the JGS graphon estimator $\widehat W^\mathrm{JGS}$ is obtained in the case, where the number of graphs $M$ is fixed. 
\begin{theorem}[Consistency of JGS graphon estimate]\label{thm:consistency-JGS}
	Let  the number of graphs $M\ge 1$ be fixed.  
	Assume that the following conditions hold.
	\begin{enumerate}[label=(\roman*)]
		\item The graphon $W$ satisfies the strict monotonicity of degrees Condition~\ref{strict} and is $L$–Lipschitz on $[0,1]^2$. Its normalized degree function $g$ is strictly increasing and satisfies the lower Lipschitz condition with constant $L_1$.
		
		\item The number of bins $k=k(N)$ satisfies the \emph{no empty blocks} condition
		\begin{equation}\label{eq:no-empty}
			k = o\!\left(\tfrac{N}{M+\log N}\right).
		\end{equation}
	\end{enumerate}
	Then, as $n_{\max}=\max_{ m\in\dblbr M} n^{(m)} \to \infty$  and $k/n_{\max}\to 0$,  the mean integrated squared error of 
	the JGS graphon estimator tends to zero
	\[
	\mathrm{MISE}\!\left(\widehat{W}^{\mathrm{JGS}}, W\right) =  \mathbb E\!\left[\big\|\widehat W^{\mathrm{JGS}}- W \big\|_{L^2}^2\right]
	\;\longrightarrow\;0.
	\]
\end{theorem}

	\begin{table*}[t!]
		\centering
		\renewcommand{\arraystretch}{1.1}
		\scriptsize
		\caption{MISE ($\times 10^{-3}$) for different graphon estimators on collections of networks with varying size.  }
		\vspace{0.3em}
		\begin{tabular}{|l|c|c|c|c|c|c|c|}
				\hline
				\textbf{ID} & SBA & SAS & USVT & SGWB & SIGL & \textbf{JGS} & \textbf{JGS-smooth} \\
				\hline
				\multicolumn{8}{|c|}{\textbf{Monotone graphons}}\\
				\hline
				1 & $37.10 \pm 0.4$ &$69.64 \pm 0.29$  &$36.95 \pm 0.4$  &$7.63 \pm 0.45$  &$1.07 \pm 0.34$  &$\boldsymbol{0.58 \pm 0.09}$ &$\boldsymbol{0.43 \pm 0.09}$  \\
				2 & $43.32 \pm 0.31$ &$55.13 \pm 0.32$ &$43.37 \pm 0.3$ &$11.23 \pm 0.33$ &$1.1 \pm 0.19$
				 &$\boldsymbol{0.82 \pm 0.07}$ &$\boldsymbol{0.64 \pm 0.06}$ \\
				3 & $104.1 \pm 0.6$&$123.6 \pm 0.6$ &$104 \pm 0.6$ &$12.42 \pm 0.33$ &$1.31 \pm 0.19$ & $\boldsymbol{0.65 \pm 0.09}$ &$\boldsymbol{0.46 \pm 0.08}$ \\
				4 &$89.66 \pm 0.75$&$114.6 \pm 0.8$ &$89.34 \pm 0.74$ &$10.15 \pm 0.74$ &$1.32 \pm 0.3$&$\boldsymbol{0.73 \pm 0.11}$ &$\boldsymbol{0.55 \pm 0.1}$ \\
				5 & $221.6 \pm 0.5$ &$204.7 \pm 0.5$ &$221.4 \pm 0.5$ &$14.09 \pm 0.35$ &$2.32 \pm 0.39$ &$\boldsymbol{0.65 \pm 0.03}$ &$\boldsymbol{0.53 \pm 0.03}$ \\
				6 & $193.7 \pm 0.3$ &$175.9 \pm 0.4$ &$193.7 \pm 0.3$ &$20.64 \pm 0.22$ &$2.42 \pm 0.83$ &$\boldsymbol{1.92 \pm 0.07}$ &$\boldsymbol{1.74 \pm 0.07}$ \\
				7  &$103.5 \pm 0.4$ &$110.6 \pm 0.4$ &$103.5 \pm 0.4$ &$21.38 \pm 0.73$ &$3.4 \pm 0.96$ &$\boldsymbol{2.64 \pm 0.1}$ & $\boldsymbol{2.51 \pm 0.09}$ \\
				8  & $83.94 \pm 0.42$&$87.25 \pm 0.49$ &$84.02 \pm 0.43$ &$13.41 \pm 0.42$ &$1.44 \pm 0.19$ &$\boldsymbol{1.11 \pm 0.19}$ & $\boldsymbol{0.89 \pm 0.06}$ \\
				9  & $40.15 \pm 0.15$&$39.19 \pm 0.19$ &$40.09 \pm 0.16$&$14.16 \pm 0.44$ &$\boldsymbol{1.06 \pm 0.23}$ &$2.34 \pm 0.07$ & $2.13 \pm 0.07$\\
				\hline
				\multicolumn{8}{|c|}{\textbf{Non-monotone graphons}}\\
				\hline
				10 & $94.64 \pm 0.14 $&$96.3 \pm 0.14$ &$94.61 \pm 0.14$ &$76.08 \pm 1.89$ &$44.48 \pm 0.63$ &$\boldsymbol{43.65 \pm 0.08}$ &$\boldsymbol{43.46 \pm 0.07}$ \\
				11 & $217 \pm 0.5 $&$213.9 \pm 0.6 $ & $217 \pm 0.5 $& $70.88 \pm 2.64 $ &$54.42 \pm 12.28 $ &$ \boldsymbol{ 43.77 \pm 0.08}$ &$\boldsymbol{ 43.63 \pm 0.08}$ \\
				12 &$ 187 \pm 2.1$ &$225.7 \pm 0.2 $ &$187.9 \pm 2.3 $ &$78.77 \pm 0.49 $ &$198.6 \pm 65.6 $ &$\boldsymbol{75.35 \pm 3.25}$ &$\boldsymbol{74.47 \pm 3.34} $ \\
				13 &$ 246.5 \pm 3.0$ & $225.5 \pm 0.2$& $253.3 \pm 2.6$& $106.5 \pm 35.1$&$197.2 \pm 37.1$& $\boldsymbol{81.6 \pm 10.3}$&$\boldsymbol{79.3 \pm 9.1}$\\
				\hline
			\end{tabular}%
		\label{table_mise_random_n}
	
	\end{table*}

\begin{figure*}[t!]
	\begin{minipage}{0.48\textwidth}
		\centering
		\includegraphics[height=5cm]{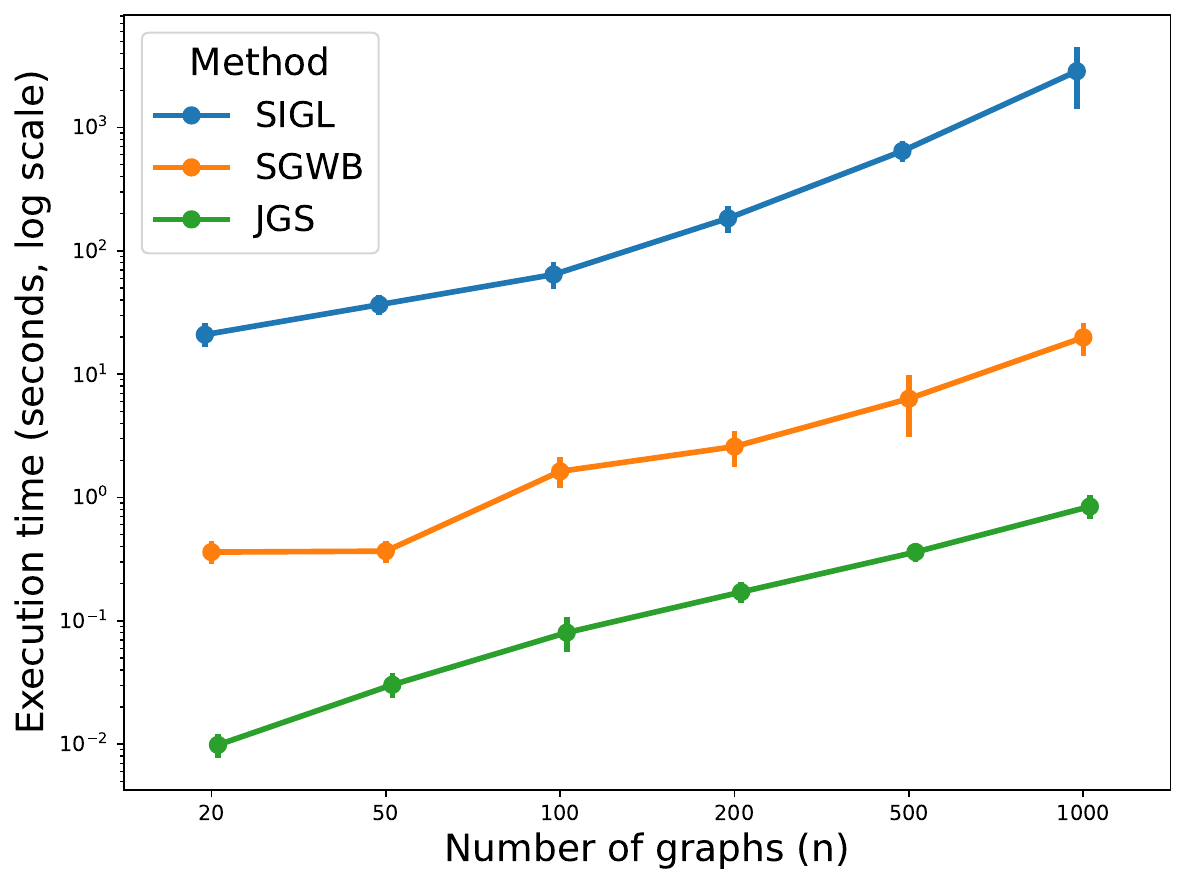}
		\caption{Time as a function of the graph size~$n$}\label{fig:time_vs_n}
	\end{minipage}
	\begin{minipage}{0.48\textwidth}
		\centering
		\includegraphics[height=5cm]{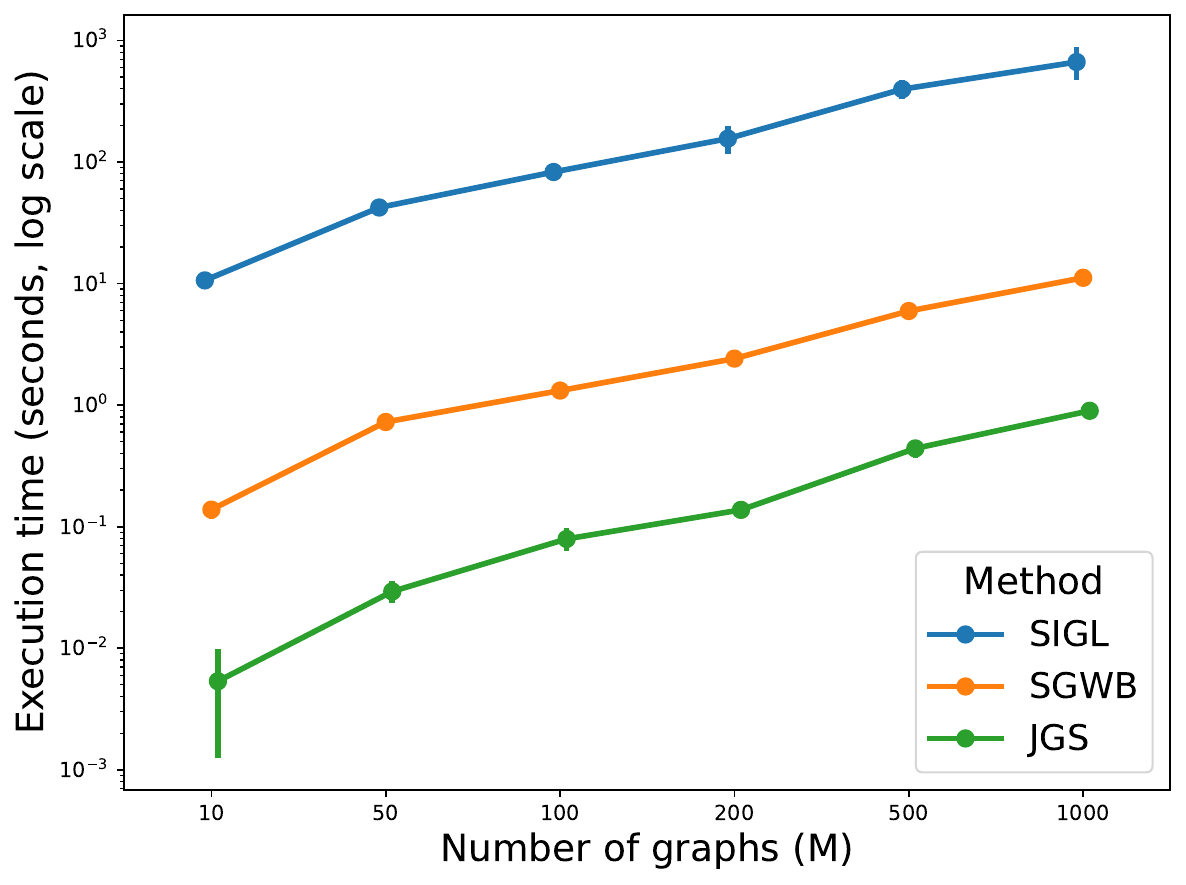}
		\caption{Time depending on the number of graphs~$M$}
		\label{fig:time_vs_M}
	\end{minipage}	
\end{figure*}

\section{NUMERICAL EXPERIMENTS}\label{sec:numerical}

In this section, the performance of   JGS  is assessed in several numerical experiments. First,  the JGS graphon estimator is compared to the state of the art on synthetic data. Then, the  behavior of the graphon estimator and the latent positions' estimators are examined in two  different asymptotic regimes. Finally, we illustrate the gain of using JGS in a classification task involving graph neural networks, on several real-world datasets.
	

\subsection{Comparison of JGS with the state of the art}\label{perfComparison}

In this study  \textit{JGS} refers to  JGS histogram estimator $\widehat{W}^{\mathrm{JGS}}$ of the graphon, where the number of blocks $k$ is chosen according to~\eqref{eq:kopt}, and \textit{JGS-smooth} to JGS combined with the total-variation–based smoothing algorithm of \cite{Chambolle2004}.  Both methods are compared to the following existing methods:  the stochastic block approximation \textit{SBA} \citep{Airoldi2013}, sorting-and-smoothing \textit{SAS} \citep{Chan2014}, universal singular value thresholding \textit{USVT} \citep{Chatterjee2015}, smoothed Gromov–Wasserstein barycenters \textit{SGWB} \citep{Xu2021}, and scalable implicit graphon learning \textit{SIGL} \citep{Azizpour2025}. 
	For SBA, SAS, USVT, and SGWB we adopt the multiple-graph implementations provided by \cite{Xu2021}, which combine single-graph estimates via pooling across networks. For SIGL, we use the authors’ publicly available code \citep{Azizpour2025}. Details on the algorithms  are provides in the supplementary material.

	Data are generated from the same thirteen  graphons as those considered in \cite{Azizpour2025}, see Table~\ref{table_mise_random_n}.
	The non-monotone graphons 10--13 do not verify the identifiability Condition~\ref{strict} and  are used to investigate the estimators' robustness.
	For each experiment, we simulate a collection of $M = 200$  networks,  with network sizes $n^{(m)}$ randomly chosen between $10$ and $100$. To evaluate the estimation accuracy  the mean   integrated   squared   error (MISE) is computed for all methods based on $20$ independent trials.

  Compared to the state of the art, JGS consistently achieves  lower MISE for all  graphons except for graphon~9, see Table~\ref{table_mise_random_n}. Unlike the other graphons, graphon 9 produces  very sparse networks and SIGL performs better than JGS in this case. This may indicate that JGS is not the optimal approach for very sparse networks, while still having the second best behavior among all methods. 
Concerning  the smoothed version, JGS-smooth, a further reduction of the  estimation error is observed. 
  Although JGS relies on the strict monotonicity condition, its performance on the non-monotone graphons (10--13) is particularly noteworthy. These cases are clearly challenging for all methods, as reflected by the substantially larger errors compared to the monotone graphons.
  
  Another advantage of JGS is its speed. As shown in Figures~\ref{fig:time_vs_n} and~\ref{fig:time_vs_M}, its execution time remains short even when the number of graphs $M$ or the network sizes $n^{(m)}$ are large. Compared to the alternative methods SIGL and SGWB, which are the most accurate approaches in the literature, the computing time of JGS is one or two orders of magnitude faster. All experiments were conducted on the same machine to ensure a fair comparison.
Altogether, these results confirm that JGS is not only more effective at leveraging information from multiple networks than current methods, but also substantially faster in terms of execution time.

\subsection{Scaling behavior of the JGS}
\begin{figure*}[t!]
	\begin{minipage}{0.48\textwidth}
		\centering
		\includegraphics[height=4cm]{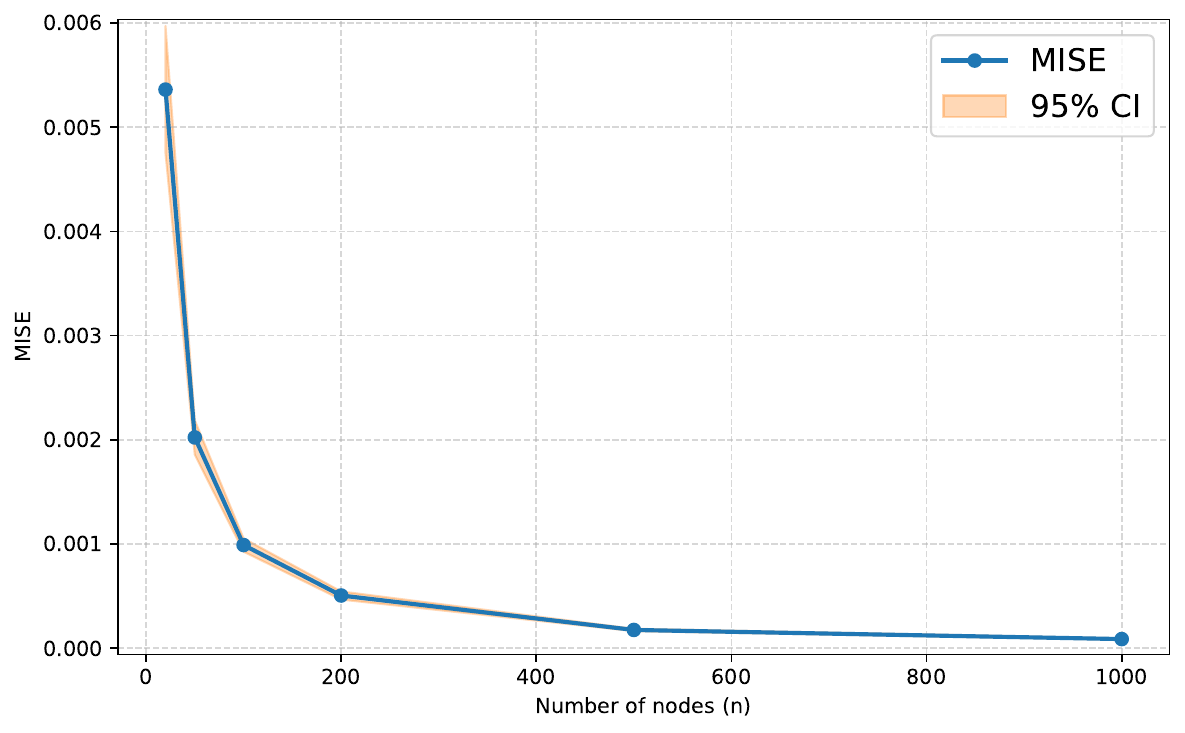}
		\caption{MISE as a function of the graph size~$n$}\label{fig:mise_vs_n}
	\end{minipage}
	\begin{minipage}{0.48\textwidth}
		\centering
		\includegraphics[height=4
		cm]{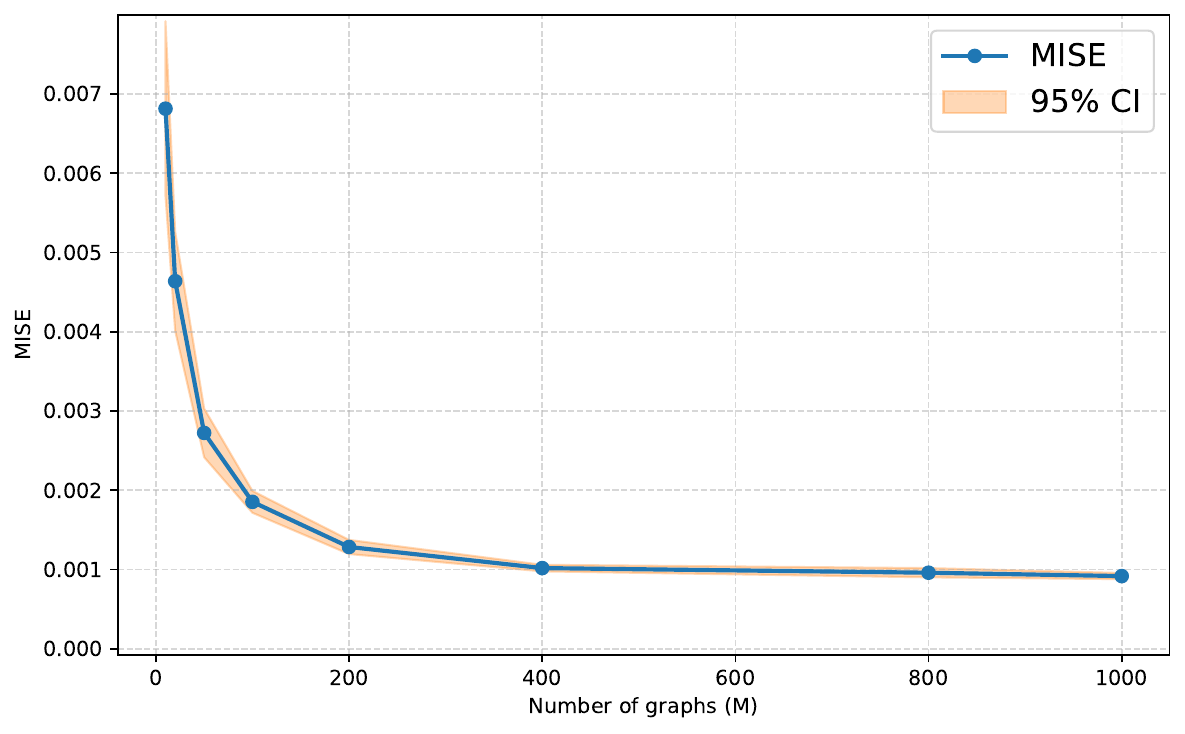}
		\caption{MISE depending on the number of graphs~$M$}
		\label{fig:mise_vs_M}
	\end{minipage}	
\end{figure*}

\begin{figure*}[t!]			 
	\begin{minipage}[b]{0.45\linewidth}	
		\centering
		\includegraphics[height=4cm]{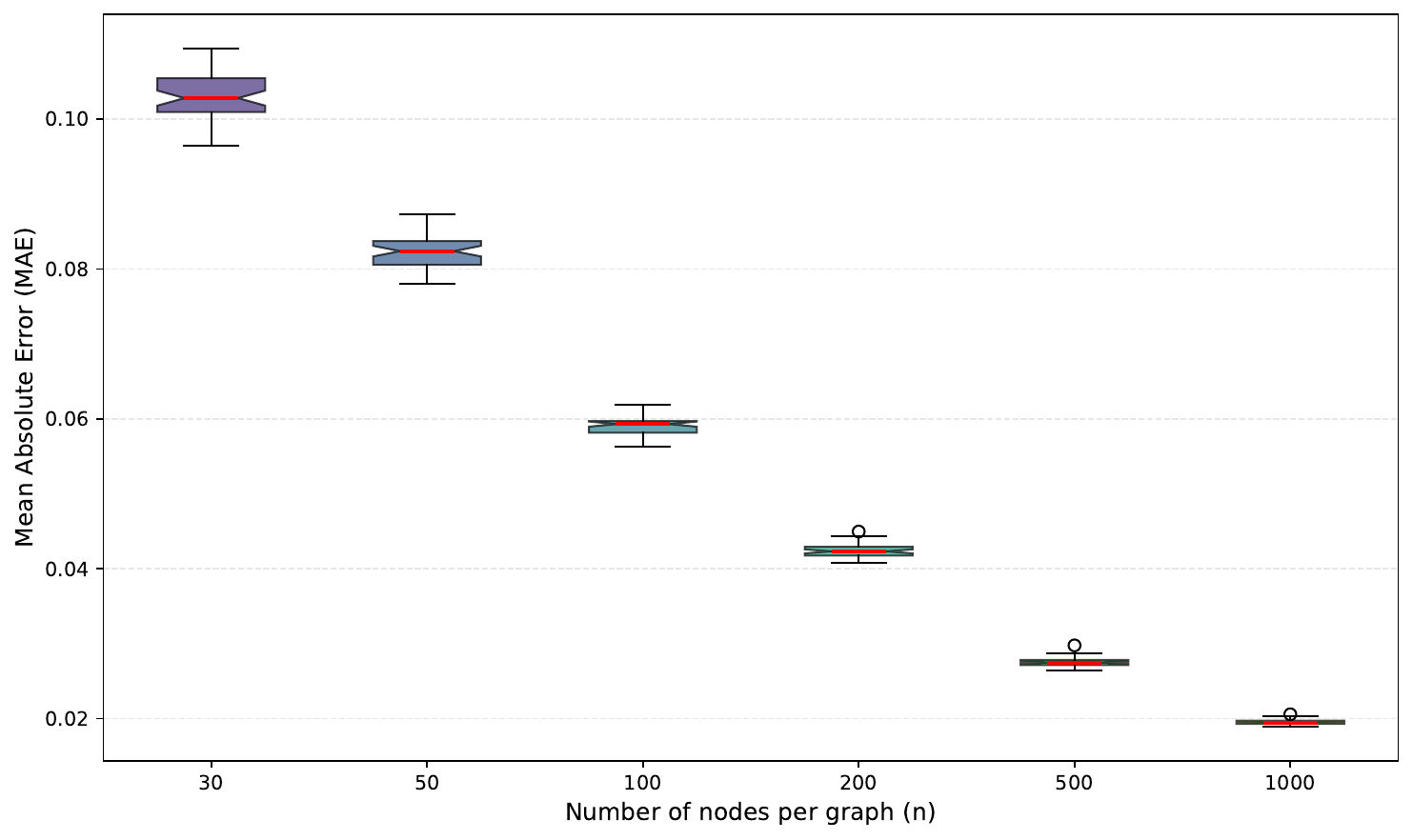}
		\caption{MAE of  the latent positions'  as the number of nodes $n^{(m)}$ increases}
		\label{latent_n}
	\end{minipage}\hfill
	\begin{minipage}[b]{0.45\linewidth}	
		\centering
		\includegraphics[height=4cm]{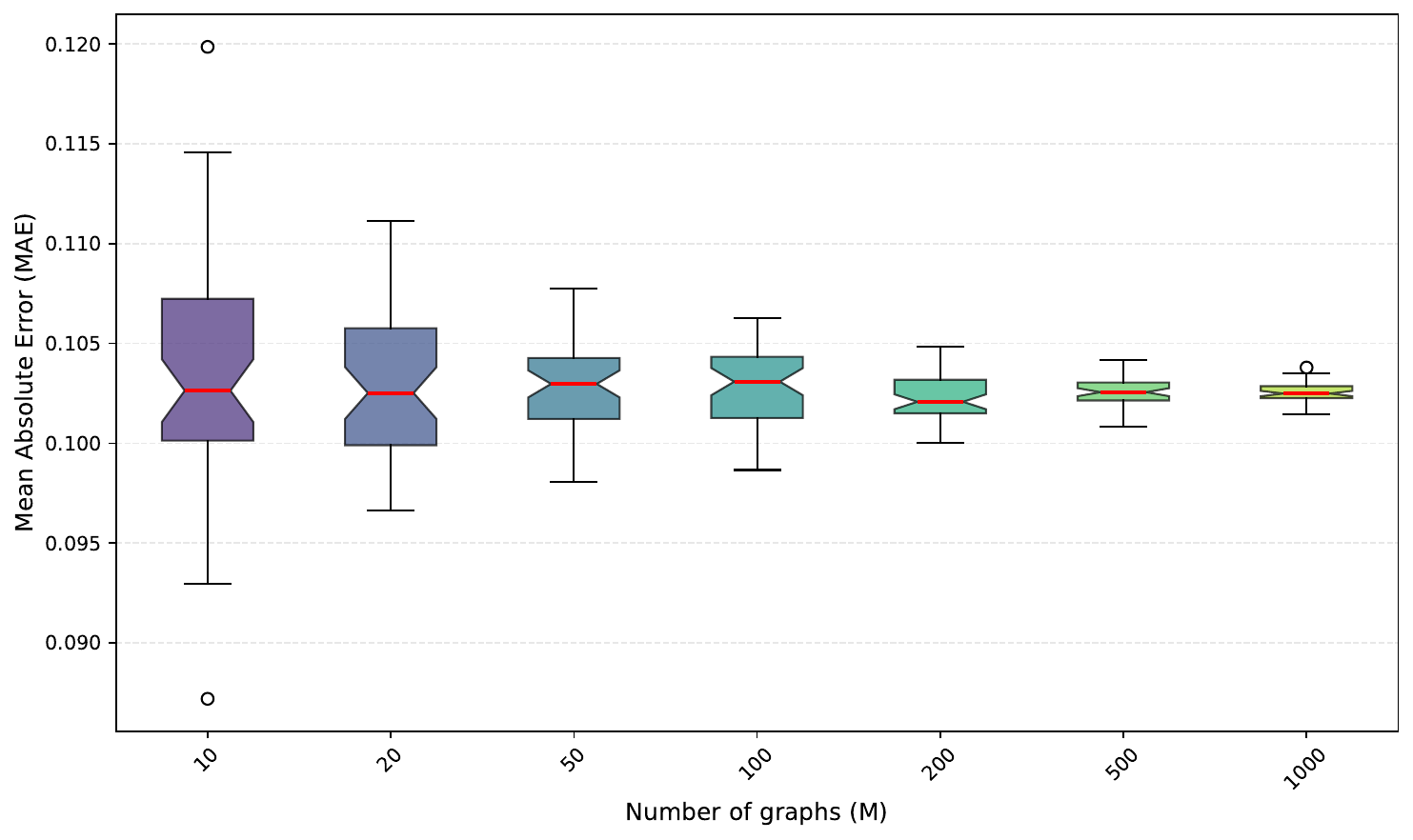}
		\caption{MAE of  the latent positions'  as the number of  graphs $M$ increases}
		\label{latent_M}
	\end{minipage}
\end{figure*}
The behaviour of the JGS estimator is studied  under two distinct asymptotic regimes, each reflecting a different practical scenario for network collections.

In the first regime,   commonly studied in the literature, the number of networks is fixed, here at \(M = 30\), while the sizes $n^{(m)}$ of all networks  increase, here from \(30\) to \(1000\) nodes. Data are generated from the graphon \(W(u,v) = uv\). As shown in Figure~\ref{fig:mise_vs_n}, the MISE decreases steadily and approaches zero as the network size $n^{(m)}$ grows. This behavior aligns with the consistency result established in Theorem~\ref{thm:consistency-JGS} and confirms that the JGS graphon estimator efficiently utilises  larger networks to improve estimation accuracy.

	In the second regime, which has received much less  attention in the literature,  all network sizes $n^{(m)}$ are fixed, here at  \(n^{(m)} = 30\),  but the number of graphs \(M\) is steadily increased, here from \(10\) to \(1000\). In this setting, Figure~\ref{fig:mise_vs_M} shows that although the MISE decreases as \(M\) grows, it converges to some non-zero value. This indicates that the JGS graphon estimator is not consistent in this regime, a limitation not specific to our method, but inherent to the fixed-graph-size asymptotic framework.

	To better understand this behavior, we examine the estimation quality of the JGS latent positions \(\hat{U}^{\mathrm{JGS}}_{i,(m)}\). As shown in Figure~\ref{latent_n}, in the first regime, these estimates are consistent as the number of nodes increases. In contrast, Figure~\ref{latent_M} reveals that in the second regime, the mean absolute error (MAE)  of the latent positions does not converge to zero. This occurs because adding more graphs   does not provide additional information about the latent positions of nodes in existing graphs. Indeed, in the first regime, the empirical normalized degrees converge to their theoretical counterpart. However, this does not happen in the second regime, where the estimated empirical degrees remain constant. 
	Consequently, this regime is intrinsically more challenging than the classical setting, where  graph sizes tend to infinity.

\subsection{Application to a classification problem on real-world data}

\begin{table}[t!]
\centering
\caption{Comparison of the classification test accuracy (\%) on IMDB-BINARY and IMDB-MULTI.}
\renewcommand{\arraystretch}{1.2}
\setlength{\tabcolsep}{8pt}
\scriptsize
\begin{tabular}{|c|c|c|}
\hline
\textbf{Method} & \textbf{IMDB-BINARY} & \textbf{IMDB-MULTI} \\ \hline
Vanilla & $73.0 \pm 3.52$ & $47.49 \pm 1.27$ \\ \hline
G-Mixup / USVT & $72.20 \pm 2.57$ & $48.99 \pm 3.01$ \\ \hline
G-Mixup / SGWB & $72.85 \pm 2.50$ & $49.67 \pm 2.41$ \\ \hline
G-Mixup / SIGL & $72.60 \pm 3.38$ & $48.63 \pm 2.27$ \\ \hline
G-Mixup / JGS & $\mathbf{75.72 \pm 2.81}$ & $\mathbf{50.23 \pm 2.15}$ \\ \hline
\end{tabular}
\label{result_real}
\end{table}
To assess the performance of     JGS   on real-world data, we integrate JGS into the G-Mixup framework \citep{Han2022,Navarro2023}, a data augmentation approach for graph classification. G-Mixup uses a graphon estimate, which is evaluated on training data, to generate new synthetic graphs by interpolation. The aim is to enrich the training set to improve the model's generalization ability.

We conducted experiments on two widely used benchmark datasets: IMDB-BINARY and IMDB-MULTI \citep{Morris2020,Yanardag2015}. The IMDB-BINARY dataset consists of 1000 graphs representing actor relationships in movies, divided into two classes. IMDB-MULTI contains 1500 networks divided into three classes. The number of nodes of the graphs range from 7 to 89, which is suitable for the JGS approach.

For the classification task, the Graph Isomorphism Network (GIN) architecture \citep{Xu2019} is applied following the training protocol described in \citep{Han2022}. Five variants are compared: a baseline model (vanilla) without data augmentation,
and G-Mixup using the USVT, SGWB, SIGL, and JGS estimators. More details on the training parameters are provided in the supplement material.

As shown in Table~\ref{result_real}, integrating our JGS method into G-Mixup leads to higher accuracy than the other variants on both IMDB-BINARY and IMDB-MULTI. This improvement is due to JGS’s ability to efficiently leverage information from multiple graphs, making it well-suited for graphon inference in heterogeneous multi-network scenarios.

\section{CONCLUSION}
This work confirms that in the multiple network setting the averaging of graphon estimates, that were conceived for single networks, is not the optimal strategy. More consistent estimators are obtained by  analyzing the entire collection of networks simultaneously as done by JGS. Furthermore, we demonstrated that such strategies do not have to be computationally  intensive;  our JGS approach achieves highly accurate results  with negligible computing times compared to methods such as SIGL and SGWB.

The following open questions have arisen during the work. Firstly, JGS appears to  perform  less well on collections of very sparse networks. The question is how to improve JGS in this case, specifically how to obtain better estimates  of the nodes' latent positions. Secondly, it has become clear that, in the asymptotic regime where the number of networks $M$ tends to infinity while the network sizes $n^{(m)}$ are bounded for all networks,  the latent position estimates are no longer consistent. This indicates that the JGS graphon estimator is not consistent in this regime. This limitation is related to the fact that the indicators used to estimate the latent positions, such as normalized degrees in our approach, cannot be consistently estimated when graph sizes remain bounded. Whether a consistent graphon estimator exists in this asymptotic framework remains an open question.





\bibliography{biblioAISTAT}

 \section*{Checklist}

\begin{enumerate}

  \item For all models and algorithms presented, check if you include:
  \begin{enumerate}
    \item A clear description of the mathematical setting, assumptions, algorithm, and/or model. Yes
    \item An analysis of the properties and complexity (time, space, sample size) of any algorithm. Yes
    \item (Optional) Anonymized source code, with specification of all dependencies, including external libraries. Yes
  \end{enumerate}

  \item For any theoretical claim, check if you include:
  \begin{enumerate}
    \item Statements of the full set of assumptions of all theoretical results. Yes
    \item Complete proofs of all theoretical results. Yes
    \item Clear explanations of any assumptions. Yes     
  \end{enumerate}

  \item For all figures and tables that present empirical results, check if you include:
  \begin{enumerate}
    \item The code, data, and instructions needed to reproduce the main experimental results (either in the supplemental material or as a URL). Yes
    \item All the training details (e.g., data splits, hyperparameters, how they were chosen). Yes
    \item A clear definition of the specific measure or statistics and error bars (e.g., with respect to the random seed after running experiments multiple times). Yes
    \item A description of the computing infrastructure used. (e.g., type of GPUs, internal cluster, or cloud provider). Yes
  \end{enumerate}

  \item If you are using existing assets (e.g., code, data, models) or curating/releasing new assets, check if you include:
  \begin{enumerate}
    \item Citations of the creator If your work uses existing assets. Yes
    \item The license information of the assets, if applicable. Yes
    \item New assets either in the supplemental material or as a URL, if applicable. Yes
    \item Information about consent from data providers/curators. Not Applicable
    \item Discussion of sensible content if applicable, e.g., personally identifiable information or offensive content. Not Applicable
  \end{enumerate}

  \item If you used crowdsourcing or conducted research with human subjects, check if you include:
  \begin{enumerate}
    \item The full text of instructions given to participants and screenshots. Not Applicable
    \item Descriptions of potential participant risks, with links to Institutional Review Board (IRB) approvals if applicable. Not Applicable
    \item The estimated hourly wage paid to participants and the total amount spent on participant compensation. Not Applicable
  \end{enumerate}

\end{enumerate}

\clearpage
\appendix
\thispagestyle{empty}

\onecolumn
\aistatstitle{Low-Complexity and Consistent Graphon Estimation\\ from Multiple Networks:
Supplementary Materials}

\section{PROOFS OF THEORETICAL RESULTS}
We first recall  the main notations used in the proofs.
	
\begin{definition}[Notations]\label{def:notation}
\begin{enumerate}[label=(\roman*)]
		\item \textbf{Graphon.}  
		The true underlying graphon is a measurable symmetric function \(	W:[0,1]^2 \to [0,1],\) assumed to be $L$–Lipschitz, that is,
		\[
		|W(u,v)-W(u',v')| \;\le\; L\big(|u-u'|+|v-v'|\big),
		\qquad \forall (u,v),(u',v')\in[0,1]^2.
		\]
		
		\item \textbf{Partition.}  
		Divide $[0,1]$ into $k$ equal subintervals
		\(
		I_s = \Big[\frac{s-1}{k},\frac{s}{k}\Big), s=1,\dots,k.
		\)
		For each graph $G^{(m)}$, let
		\[
		J_s^{(m)} = \left\{\, i : U_{i,(m)} \in I_s \,\right\}, 
		\qquad 
		\widehat J_s^{(m)} = \left\{\, i : \widehat U_{i,(m)}^{\mathrm{JGS}} \in I_s \,\right\}.
		\]
		be the set of nodes of network $G^{(m)}$ that fall into $I_s$ according to the true latent positions $U_{i,(m)}$ and $\widehat U_{i,(m)}^{\mathrm{JGS}}$, resp.
		\item \textbf{Dyads per block.}  
		The set of dyads in block $(s,t)$ is
		\[
		\Omega_{s,t} = \bigcup_{m\in\dblbr M}
		\left\{(i,j;m): i \in J_s^{(m)},\, j \in J_t^{(m)} \right\},
		\]
		with cardinal number $M_{s,t}=|\Omega_{s,t}|$.  
		Similarly, with estimated positions,
		\[
		\widehat\Omega_{s,t} = \bigcup_{m\in\dblbr M} 
		\left\{(i,j;m): i \in \widehat J_s^{(m)},\, j \in \widehat J_t^{(m)} \right\}, 
		\quad 
		\widehat M_{s,t}=|\widehat\Omega_{s,t}|.
		\]
		
		\item \textbf{JGS graphon estimator.}  
		The JGS estimate  defined on a $k\times k$ partition of $[0,1]^2$ is given by
		$$
		\widehat W^{\mathrm{JGS}}(u,v) =\sum_{s,t\in\dblbr k^2} \widehat H_{s,t}^{\mathrm{JGS}} \mathds{1}_{(u,v) \in I_s \times I_t}, 
		$$
		with
		\[
		\widehat H_{s,t}^{\mathrm{JGS}} 
		= \frac{1}{\widehat M_{s,t}\vee 1} 
		\sum_{m\in\dblbr M}\sum_{i \in \widehat J_s^{(m)}}\sum_{j \in \widehat J_t^{(m)}} A_{i,j}^{(m)}.
		\]
		
		\item \textbf{Oracle histogram.}  
		The oracle version defined on a $k\times k$ partition of $[0,1]^2$ using the true latent positions $U_{i,(m)}$
		is
		$$
		W^{\mathrm{oracle}}(u,v) = \sum_{s,t\in\dblbr k^2} H_{s,t}\mathds{1}_{(u,v) \in I_s \times I_t}, 
		$$
		where
		\[
		H_{s,t} = \frac{1}{M_{s,t}\vee 1} 
		\sum_{m\in\dblbr M}\sum_{i \in J_s^{(m)}} \sum_{j \in J_t^{(m)}} A_{i,j}^{(m)}.
		\]
		
		\item \textbf{Block averages of $W$.}  
		For the bias–variance analysis, we define
		\[
		\Phi_{s,t} = \mathbb E \left[H_{s,t}\,\middle|\,\mathcal U\right]
		= \frac{1}{M_{s,t}} \sum_{m\in\dblbr M}\sum_{i \in J_s^{(m)}} \sum_{j \in  J_t^{(m)}} W\left(U_{i,(m)},U_{j,(m)}\right),
		\]
		and
		\[
		\widehat \Phi_{s,t} = \mathbb E \left[\widehat H_{s,t}^{\mathrm{JGS}}\,\middle|\,\mathcal U\right]
		= \mathbb E \left[
		\frac{1}{\widehat M_{s,t}} \sum_{m\in\dblbr M}\sum_{i \in \widehat J_s^{(m)}}\sum_{j \in  \widehat J_t^{(m)}} W\left(U_{i,(m)},U_{j,(m)}\right)
		\;\middle|\;\mathcal U
		\right],
		\]
		where $\mathcal{U}$ is the set of all latent positions in the collection of networks given by
		\[
		\mathcal{U} = \left\{ U_{i,(m)} : m \in \llbracket M \rrbracket,\; i \in \llbracket n^{(m)} \rrbracket \right\}.
		\]
		
		\item \textbf{Discretized oracle graphon.}  
		Define the piecewise-constant version of $W$  on a $k\times k$ partition of $[0,1]^2$ as
			$$
		W^{\mathrm{approx}}(u,v) = \sum_{s,t\in\dblbr k^2} \Phi_{s,t},\mathds{1}_{(u,v) \in I_s \times I_t}, 
		$$
		
		\item \textbf{Node ranks.}  
		For each node $(i,m)\in\dblbr n^{(m)}\times \dblbr M$, let
		$$
		\widehat r_{i,(m)} = \big(\widehat \sigma^{\mathrm{JGS}}\big)^{-1}(i,m)
		$$
		be its empirical rank among all $N=\sum_{m\in\dblbr M}n^{(m)}$ nodes.  
		The \emph{oracle rank} is
		\[
		r_{i,(m)} 
		= \left|\left\{(j,\ell): g(U_{j,(\ell)}) \le g(U_{i,(m)})\right\}\right|,
		\]
		where $g(u)=\int_0^1 W(u,v)\,dv$ is the normalized degree function.
	\end{enumerate}
\end{definition}

\subsection{Proof of Proposition ~\ref{prop:latent_consistency} }
The consistency of the latent positions' estimators $\widehat U_{i,(m)}$ depends on the consistency of the correct ordering of the nodes, that is, on the exact estimation of the rank of a node in the collection of  all nodes. This is formalized in the following lemma.

\begin{lemma}[Control of the rank discrepancy]
	\label{lemma:rank_error}
	Suppose that the normalized degree function \( g \colon [0,1] \to \mathbb{R} \) is strictly increasing and satisfies the lower Lipschitz condition, that is, there exists a constant \( L_1 > 0 \) such that for all \( x, y \in [0,1] \),
	\[
	L_1 |x - y| \le |g(x) - g(y)|.
	\]
	Then, for any fixed node \( (i,m)\in\dblbr n^{(m)}\times \dblbr M \), with probability at least \( 1 - \delta \), the discrepancy between the empirical rank \(\widehat r_{i,(m)} \) and the theoretical rank \( r_{i,(m)} \) satisfies
	\[
	\left| \frac{\widehat r_{i,(m)}}{N} - \frac{r_{i,(m)}}{N} \right| \le  \eta_i^{(m)},
	\]
	where \(\eta_i^{(m)}\) is defined in the main paper in \eqref{eq:eq_eta_im}.
\end{lemma}

\begin{proof}
    Fix a node \( (i,m) \in \llbracket n^{(m)} \rrbracket \times \llbracket M \rrbracket \). Recall that conditionally on the latent variables \( U_{i,(m)}, U_{j,(m)} \), the variables \( A_{i,j}^{(m)} \) are independent with Bernoulli distribution with parameter \( W(U_{i,(m)}, U_{j,(m)}) \). Therefore, conditionally on \( U_{i,(m)} \), we have
	\[
	\mathbb{E}\left[\widehat d_{i,(m)}^{\mathrm{norm}} \mid U_{i,(m)}\right] 
	= \frac{1}{n^{(m)}} \sum_{j\in \dblbr{n^{(m)}}} \mathbb{E}[A_{i,j}^{(m)} \mid U_{i,(m)}] 
	= \int_0^1 W(U_{i,(m)}, v)\, dv = g(U_{i,(m)}).
	\]
	Since the \( A_{i,j}^{(m)} \in \{0,1\} \) are conditionally independent,   Hoeffding’s inequality yields, conditionally on \( U_{i,(m)} \),
	\[
	\mathbb{P}\left( \left| \widehat d_{i,(m)}^{\mathrm{norm}} - g(U_{i,(m)}) \right| \ge \varepsilon_n^{(m)} \,\middle|\, U_{i,(m)} \right)
	\le 2 \exp\left( -2 n^{(m)} (\varepsilon_n^{(m)})^2 \right).
	\]
	Setting
	\[
	\varepsilon_n^{(m)} = \sqrt{ \frac{ \log(2/\delta) }{2 n^{(m)} } },
	\]
	we obtain
	\[
	\mathbb{P}\left( \left| \widehat d_{i,(m)}^{\mathrm{norm}} - g(U_{i,(m)}) \right| \ge \varepsilon_n^{(m)} \,\middle|\, U_{i,(m)} \right) \le \delta.
	\]
	We now define the following set of potentially misranked nodes
	\[
	C_i^{(m)} := \left\{ (j,\ell) \neq (i,m) : \left| g(U_{j,(\ell)}) - g(U_{i,(m)}) \right| \le \varepsilon_i^{(m)} + \varepsilon_j^{(\ell)} \right\}.
	\]
	Only the nodes in \( C_i^{(m)}  \) can possibly alter the relative ordering between \( \widehat d_{i,(m)}^{\mathrm{norm}} \) and \( g(U_{i,(m)}) \), since for all other pairs, the deviation is too large so that a possible reversal has negligible probability. 
	Indeed, suppose that
	\[
	g(U_{j,(\ell)}) < g(U_{i,(m)}) - \varepsilon_i^{(m)} - \varepsilon_j^{(\ell)}.
	\]
	Then, by the concentration bounds
	\[
	\widehat d_{j,(\ell)}^{\mathrm{norm}} \le g(U_{j,(\ell)}) + \varepsilon_j^{(\ell)} 
	< g(U_{i,(m)}) - \varepsilon_i^{(m)} \le \widehat d_{i,(m)}^{\mathrm{norm}},
	\]
that is	$\widehat d_{j,(\ell)}^{\mathrm{norm}} < \widehat d_{i,(m)}^{\mathrm{norm}}$. 
	Similarly, if
	\[
	g(U_{j,(\ell)}) > g(U_{i,(m)}) + \varepsilon_i^{(m)} + \varepsilon_j^{(\ell)},
	\]
	then
	\[
	\widehat d_{j,(\ell)}^{\mathrm{norm}} \ge g(U_{j,(\ell)}) - \varepsilon_j^{(\ell)} 
	> g(U_{i,(m)}) + \varepsilon_i^{(m)} \ge \widehat d_{i,(m)}^{\mathrm{norm}},
	\]
	that is $\widehat d_{j,(\ell)}^{\mathrm{norm}} > \widehat d_{i,(m)}^{\mathrm{norm}}$.
	Thus, only the nodes \( (j,\ell) \in C_i^{(m)} \) may provoke a change in the empirical rank of node \( (i,m) \), while all others preserve the ordering. This justifies restricting the rank discrepancy analysis to the set \( C_i^{(m)} \).
	We thus obtain
	\[
	\left| \frac{\widehat r_{i,(m)}}{N} - \frac{r_{i,(m)}}{N} \right| \le \frac{|C_i^{(m)}|}{N}.
	\]
	Let us now control \( \mathbb{E}[|C_i^{(m)}| \mid U_{i,(m)}] \). By definition,
	\[
	\mathbb{E}\left[|C_i^{(m)}| \mid U_{i,(m)}\right] = \sum_{(j,\ell) \ne (i,m)} \mathbb{P}\left( (j,\ell) \in C_i^{(m)} \,\middle|\, U_{i,(m)} \right).
	\]
	Using the bi-Lipschitz property of \( g \), we know that
	\[
	\left| g(U_{j,(\ell)}) - g(U_{i,(m)}) \right| \le \varepsilon_i^{(m)} + \varepsilon_j^{(\ell)}.\]
	Thus, \[
	\left| U_{j,(\ell)} - U_{i,(m)} \right| \le \frac{ \varepsilon_i^{(m)} + \varepsilon_j^{(\ell)} }{L_1}.
	\]
	Since \( U_{j,(\ell)} \sim \text{Uniform}[0,1] \)   is independent of \( U_{i,(m)} \), conditionally on \( U_{i,(m)} \), we get
	\[
	\mathbb{P}\left( (j,\ell) \in C_i^{(m)} \,\middle|\, U_{i,(m)} \right)
	\le 1 \wedge  \left( 2 \cdot \frac{ \varepsilon_i^{(m)} + \varepsilon_j^{(\ell)} }{L_1} \right).
	\]
	Therefore
	\[
	\mathbb{E}\left[|C_i^{(m)}| \mid U_{i,(m)}\right] \le  N \wedge \left\{ \frac{2}{L_1} \left( (N-1) \varepsilon_i^{(m)} + \sum_{(j,\ell) \ne (i,m)} \varepsilon_j^{(\ell)} \right) \right\}.
	\]
	To get a high-probability bound on \( |C_i^{(m)}| \), we use Hoeffding's inequality conditionally on \( U_{i,(m)}\), since \( |C_i^{(m)}| \) is a sum of independent indicators, for any $\eta>0$,
	\[
	\mathbb{P}\left( \left| \frac{|C_i^{(m)}|}{N} - \frac{ \mathbb{E}[|C_i^{(m)}| \mid U_{i,(m)}] }{N} \right| \ge \eta \,\middle|\, U_{i,(m)} \right)
	\le 2 \exp\left( -2 N \eta^2 \right).
	\]
	By choosing \( \eta = \sqrt{ \frac{ \log(2/\delta) }{2N} } \), we get
	\[
	\mathbb{P}\left( \left| \frac{|C_i^{(m)}|}{N} - \frac{ \mathbb{E}[|C_i^{(m)}| \mid U_{i,(m)}] }{N} \right| \ge \sqrt{ \frac{ \log(2/\delta) }{2N} } \,\middle|\, U_{i,(m)} \right) \le \delta.
	\]
	Finally, with probability at least \( 1 - \delta \), we obtain
	\[
	\left| \frac{\widehat r_{i,(m)}}{N} - \frac{r_{i,(m)}}{N} \right| \le  1 \wedge \left\{  \frac{2}{L_1} \left( \varepsilon_i^{(m)} + \frac{1}{N} \sum_{(j,\ell) \ne (i,m)} \varepsilon_j^{(\ell)} \right) \right\}
	+ \sqrt{ \frac{ \log(2/\delta) }{2N} }.
	\]
\end{proof}

\vspace{1cm}
\begin{proof}[Proof of Proposition~\ref{prop:latent_consistency}] 
	Fix a node \( (i,m) \in \llbracket n^{(m)} \rrbracket \times \llbracket M \rrbracket \). 
	The estimation error of \( \widehat{U}_{i,(m)}^{\mathrm{JGS}} \) 
	is bounded by a sum of three terms
	\begin{equation}\label{eqbound3terms}
	\left| \widehat{U}_{i,(m)}^{\mathrm{JGS}} - U_{i,(m)} \right| 
	\le 
	\left| \widehat{U}_{i,(m)}^{\mathrm{JGS}} - \frac{\widehat r_{i,(m)}}{N} \right| 
	+ \left| \frac{\widehat r_{i,(m)}}{N} - \frac{r_{i,(m)}}{N} \right| 
	+ \left| \frac{r_{i,(m)}}{N} - U_{i,(m)} \right|.
	\end{equation}
	For the first term, by definition of the estimator \( \widehat{U}_{i,(m)}^{\mathrm{JGS}} = \frac{\widehat r_{i,(m)}}{N} - \frac{1}{2N} \), we immediately get
		\[
		\left| \widehat{U}_{i,(m)}^{\mathrm{JGS}} - \frac{\widehat r_{i,(m)}}{N} \right| = \frac{1}{2N}.
		\]
The second term  in \eqref{eqbound3terms} 
is bounded by using Lemma \ref{lemma:rank_error}.
		
The third term \( \left| \frac{r_{i,(m)}}{N} - U_{i,(m)} \right| \)   corresponds to the error between the empirical and theoretical quantiles of \( g(U_{i,(m)}) \). 
		Since the normalized degree function \( g \) is strictly increasing and \( U \sim \mathrm{Unif}[0,1] \),  
		\[
		F_g(g(U_{i,(m)})) = \mathbb{P}(g(U) \le g(U_{i,(m)})) = \mathbb{P}(U \le U_{i,(m)}) = U_{i,(m)}.
		\]
		On the other hand, we have
		\[
		\widehat{F}_g(g(U_{i,(m)})) = \frac{r_{i,(m)}}{N},
		\]
		which is the empirical CDF of \( g(U) \), using all \(N\) nodes in the collection,   evaluated at \( g(U_{i,(m)}) \). 
		By the Dvoretzky–Kiefer–Wolfowitz (DKW) inequality, for any \( \delta \in (0,1) \), we have with probability at least \( 1 - \delta \),
		\[
		\sup_{x \in \mathbb{R}} \left| \widehat{F}_g(x) - F_g(x) \right| \le \sqrt{ \frac{ \log(2/\delta) }{2N} }.
		\]
		In particular, this gives
		\[
		\left| \frac{r_{i,(m)}}{N} - U_{i,(m)} \right| 
		= \left| \widehat{F}_g(g(U_i^{(m)})) - F_g(g(U_{i,(m)})) \right|
		\le \sqrt{ \frac{ \log(2/\delta) }{2N} }.
		\]
	
Combining the three terms, we obtain
	\[
	\left| \widehat{U}_{i,(m)}^{\mathrm{JGS}} - U_{i,(m)} \right| 
	\le 
	\frac{1}{2N}
	+  1 \wedge \left\{  \frac{2}{L_1} \left( \varepsilon_i^{(m)} + \frac{1}{N} \sum_{(j,\ell) \ne (i,m)} \varepsilon_j^{(\ell)} \right) \right\}
	+ 2 \sqrt{ \frac{ \log(2/\delta) }{2N} }.
	\]	 	
\end{proof}

\subsection{Proof of Theorem~\ref{thm:consistency-JGS}}
The proof of the theorem is structured around three auxiliary lemmas:  
\textit{(i)} the first lemma controls the discretization error between the true graphon $W$ and its block approximation $W^{\mathrm{approx}}$,  
\textit{(ii)} the second lemma bounds the error between the oracle histogram $W^{\mathrm{oracle}}$
and the block approximation of the true graphon $W^{\mathrm{approx}}$,  
and \textit{(iii)} the third lemma establishes the stochastic error between the JGS estimator
$\widehat W^{\mathrm{JGS}}$ and the oracle graphon $W^{\mathrm{oracle}}$, showing how vertex misclassification propagates to the final error.  

\begin{lemma}[Discretization error]\label{lem:discretization_error}
	Assume that the true graphon $W:[0,1]^2 \to [0,1]$ is  $L$--Lipschitz. Let $W^{\mathrm {approx}}$ denote its blockwise discretization on a $k\times k$ partition of $[0,1]^2$.  
	Then there exists a constant $C>0$, independent of $N$, such that
	\[
	\mathbb E \left[\|W^{\mathrm {approx}} - W\|_{L^2}^2\right]
	\;\le\; \frac{C\,L^2}{k^2}.
	\]
\end{lemma}

The proof of this lemma is standard in the literature on graphon models and histogram approximations (see Lemma 1 in \cite{Chan2014}).  

\begin{lemma}[Oracle vs.~block approximation]\label{lem:histogram_true_latent}
	Let $W^{\mathrm{oracle}}$ denote the oracle histogram estimator
	with blockwise discretization on a $k\times k$ partition of $[0,1]^2$, and let $W^{\mathrm {approx}}$ denote the block approximation of the true graphon. Then
	\[
	\mathbb E \left[\big\| W^{\mathrm{oracle}}- W^{\mathrm {approx}}\big\|_{L^2}^2\right]
	= \frac{1}{k^2}\sum_{s,t\in\dblbr k^2} \mathbb E \left[(H_{s,t}-\Phi_{s,t})^2\right]
	\;\le\;\frac{1}{4k^2}\sum_{s,t\in\dblbr k^2} \mathbb E \Big[\frac{1}{M_{s,t}}\Big],
	\]
	where $\Phi_{s,t}$ is the blockwise expectation of $W$ on $I_s\times I_t$ and $M_{s,t}$ is the number of observed dyads in block $(s,t)$.
\end{lemma}

\begin{proof}
	Fix $(s,t)\in\dblbr k^2$. Conditional on the latent positions $\mathcal U$, we have
	\[
	\Phi_{s,t} = \mathbb E \left[H_{s,t}\,\middle|\,\mathcal U\right]
	= \frac{1}{M_{s,t}} \sum_{m\in\dblbr M}\sum_{i \in J_s^{(m)}} \sum_{j \in  J_t^{(m)}} W\left(U_{i,(m)},U_{j,(m)}\right),
	\]
	Therefore
	\begin{align*}
		\mathbb E \big[(H_{s,t}-\Phi_{s,t})^2 \,\big|\, \mathcal U \big]
		&= \mathrm{Var}\left(\frac{1}{M_{s,t}} \sum_{m\in\dblbr M}\sum_{i \in J_s^{(m)}}\sum_{j \in  J_t^{(m)}} A^{(m)}_{i,j}\mid U_{i,(m)},U_{j,(m)}\,\Big|\,\mathcal U \right) \\
		&\leq \frac{1}{M_{s,t}^2}\sum_{m\in\dblbr M}\sum_{i \in J_s^{(m)}}\sum_{j \in  J_t^{(m)}}\mathrm{Var}\left(A^{(m)}_{i,j}\mid U_{i,(m)},U_{j,(m)}\right).
	\end{align*}
	Since $A^{(m)}_{i,j}|\mathcal U\sim \mathrm{Bernoulli}(W\left(U_{i,(m)},U_{j,(m)}\right)$, 
	\[
	\mathrm{Var}\left((A_{i,j}^{(m)}\mid U_{i,(m)},U_{j,(m)}\right)=W\left(U_{i,(m)},U_{j,(m)}\right)\big(1-W\left(U_{i,(m)},U_{j,(m)}\right)\big)\le \frac{1}{4}.
	\]
	Hence
	\[
	\mathbb E \big[(H_{s,t}-\Phi_{s,t})^2 \,\big|\, \mathcal U\big]
	\;\le\;\frac{1}{M_{s,t}^2}\cdot M_{s,t}\cdot \frac{1}{4}
	= \frac{1}{4M_{s,t}}.
	\]
	Taking the expectation over the latent positions gives
	\[
	\mathbb E \left[(H_{s,t}-\Phi_{s,t})^2\right]\;\le\;\mathbb E \Big[\frac{1}{4M_{s,t}}\Big].
	\]
This yields
	\[
	\mathbb E \left[\big\| W^{\mathrm{oracle}}- W^{\mathrm {approx}}\big\|_{L^2}^2\right]
	=\frac{1}{k^2}\sum_{s,t\in\dblbr k^2} \mathbb E \left[(H_{s,t}-\Phi_{s,t})^2\right]
	\;\le\;\frac{1}{4k^2}\sum_{s,t\in\dblbr k^2} \mathbb E \Big[\frac{1}{M_{s,t}}\Big].
	\]
\end{proof}

\begin{lemma}[Stochastic error]\label{lem:stochastic_error}
	Assume that for all vertices $(i,m)\in\dblbr{ n^{(m)}}\times \dblbr M$ the latent position errors satisfy
	\[
\left|\widehat U_{i,(m)}^{\mathrm{JGS}}-U_{i,(m)}\right| \; \leq \; \eta_i^{(m)},
\]
 where \(\eta_i^{(m)}\) is defined in the main paper in \eqref{eq:eq_eta_im}.
	Then
	\[
	\mathbb E \left[\big\|\widehat W^{\mathrm{JGS}}- W^{\mathrm{oracle}}\big\|_{L^2}^2\right]
	\;\le\;\frac{1}{k^2}\sum_{s,t\in\dblbr k^2} \mathbb E \left[\frac{3}{4\widehat M_{s,t}}
	+ \frac{3}{4M_{s,t}}
	+ \frac{48}{M_{s,t}^2}\,
	\Bigg(\sum_{m\in\dblbr M}\big(|J_t^{(m)}|\,|B_s^{(m)}|
	+|J_s^{(m)}|\,|B_t^{(m)}|\big)\Bigg)^{\!2}\right],
	\]
	where $J_r^{(m)}=\{\,i:U_{i,(m)}\in I_r\,\}$ and $M_{s,t}=|\Omega_{s,t}|$, and 
	\[
	C_s^{(m)} := \left\{ i \in \dblbr{n^{(m)}}: 
	\mathrm{dist}(U_{i,(m)},\partial I_s) \le \eta_i^{(m)}\right\}.
	\]
\end{lemma}

\begin{proof}
	We note that
	\begin{align}\label{eqLemmeStochError1}
		\mathbb E \left[\big\|\widehat W^{\mathrm{JGS}}- W^{\mathrm{oracle}}\big\|_{L^2}^2\right]
		&= \frac{1}{k^2}\sum_{s,t\in\dblbr k^2} \mathbb E \big[(\widehat H_{s,t}^{\mathrm{JGS}}- H_{s,t})^2 \big].
	\end{align}
	By the elementary inequality $(a+b+c)^2 \le 3(a^2+b^2+c^2)$, we obtain
\begin{equation}\label{eqLemmeStochError2}
	\mathbb E \left[(\widehat H_{s,t}^{\mathrm{JGS}}- H_{s,t})^2 \right]
	\;\le\; 
	3\,\mathbb E \left[(\widehat H_{s,t}^{\mathrm{JGS}}- \widehat\Phi_{s,t})^2\right] 
	+3\,\mathbb E \left[(\widehat\Phi_{s,t}- \Phi_{s,t})^2 \right] 
	+3\,\mathbb E \left[(\Phi_{s,t}- H_{s,t})^2\right]. 
\end{equation}
	For the first term, by the conditional independence of Bernoulli edges, we compute
	\begin{align}
		\mathbb E \left[(\widehat H_{s,t}^{\mathrm{JGS}}- \widehat\Phi_{s,t})^2\right]
		&= \mathbb E \left[\mathrm{Var}\left(\widehat H_{s,t}^{\mathrm{JGS}}\,\big|\,\mathcal{U}\right)\right] \notag \\
		&= \mathbb E \left[\frac{1}{\widehat M_{s,t}^2}\,
		\sum_{m\in\dblbr M} \sum_{i\in \widehat J_s^{(m)}}\sum_{j\in \widehat J_t^{(m)}} 
		\mathrm{Var}\big(A_{i,j}^{(m)}\mid U_{i,(m)}, U_{j,(m)}\big)\right] \notag \\
		&\le \mathbb E \left[\frac{1}{\widehat M_{s,t}^2}\,
		\sum_{m\in\dblbr M} \sum_{i\in \widehat J_s^{(m)}}\sum_{j\in \widehat J_t^{(m)}} 
		\frac{1}{4}\right] \notag \\
		&= \mathbb E \left[\frac{1}{4\widehat M_{s,t}}\right], \label{a}
	\end{align}
where we used    $\mathrm{Var}\left((A_{i,j}^{(m)}\mid  U_{i,(m)}, U_{j,(m)}\right)\le 1/4$.  
	Similarly, for the third term, we obtain
	\begin{align}
		\mathbb E \left[(H_{s,t}- \Phi_{s,t})^2\right]
		\le \mathbb E \left[\frac{1}{4 M_{s,t}}\right]. \label{b}
	\end{align}
	We now turn to the middle term. By definition,
	\begin{align*}
	\Phi_{s,t}
	&= \frac{1}{M_{s,t}}
	\sum_{m\in\dblbr M}\sum_{i\in J_s^{(m)}}\sum_{j\in J_t^{(m)}}
	W\left(U_{i,(m)},U_{j,(m)}\right),  \\ 
	\widehat\Phi_{s,t}
	&= \mathbb E \left[
	\frac{1}{\widehat M_{s,t}}
	\sum_{m\in\dblbr M}\sum_{i\in \widehat J_s^{(m)}}\sum_{j\in \widehat J_t^{(m)}}
	W\left(U_{i,(m)},U_{j,(m)}\right)
	\;\middle|\;\mathcal U \right].
	\end{align*}
	We add and subtract
	\[
	\frac{1}{M_{s,t}}
	\sum_{m\in\dblbr M}\sum_{i\in \widehat J_s^{(m)}}\sum_{j\in \widehat J_t^{(m)}}
	W\left(U_{i,(m)},U_{j,(m)}\right)
	\]
	inside the conditional expectation to obtain
	\begin{align*}
		\widehat\Phi_{s,t}-\Phi_{s,t}
		&= \mathbb E \left[
		\Big(\frac{1}{\widehat M_{s,t}}-\frac{1}{M_{s,t}}\Big)
		\sum_{m\in\dblbr M}\sum_{i\in \widehat J_s^{(m)}}\sum_{j\in \widehat J_t^{(m)}}
		W\left(U_{i,(m)},U_{j,(m)}\right)
		\;\middle|\;\mathcal U \right] \\
		&\quad+ \frac{1}{M_{s,t}}
		\mathbb E \left[
		\sum_{m\in\dblbr M}\Big(\sum_{i\in \widehat J_s^{(m)}}\sum_{j\in \widehat J_t^{(m)}}W\left(U_{i,(m)},U_{j,(m)}\right)
		-\sum_{i\in J_s^{(m)}}\sum_{j\in J_t^{(m)}} W\left(U_{i,(m)},U_{j,(m)}\right)\Big)
		\;\middle|\;\mathcal U \right]\\
	&= \mathbb E \left[
		\Big(\frac{1}{\widehat M_{s,t}}-\frac{1}{M_{s,t}}\Big)
		\sum_{m\in\dblbr M}\sum_{i\in \widehat J_s^{(m)}}\sum_{j\in \widehat J_t^{(m)}}
		W\left(U_{i,(m)},U_{j,(m)}\right)
		\;\middle|\;\mathcal U \right]\\
		&\quad+ \frac{1}{M_{s,t}}
		\mathbb E \left[\sum_{(i,j;m) \in D_{s,t}} W\left(U_{i,(m)},U_{j,(m)}\right) \;\middle|\;\mathcal U \right],
	\end{align*}
		where \(
	D_{s,t}=\Omega_{s,t}\,\triangle\,\widehat\Omega_{s,t}
	=\big(\widehat\Omega_{s,t}\setminus\Omega_{s,t}\big)\cup\big(\Omega_{s,t}\setminus\widehat\Omega_{s,t}\big)
	\) the symmetric difference of the index sets, with 
	\[
		\Omega_{s,t} = \bigcup_{m\in\dblbr M} 
		\{(i,j;m): i \in J_s^{(m)},\, j \in J_t^{(m)} \}, \quad \text{and} \quad 
		\widehat\Omega_{s,t} = \bigcup_{m\in\dblbr M} 
		\{(i,j;m): i \in \widehat J_s^{(m)},\, j \in \widehat J_t^{(m)} \}.
		\]
	Since $0\le W \le 1$, we obtain the bound
	\begin{align*}
		|\widehat\Phi_{s,t}-\Phi_{s,t}| 
		&\leq 
		\mathbb E \left[
		\left|\frac{1}{\widehat M_{s,t}}-\frac{1}{M_{s,t}}\right|
		\sum_{m\in\dblbr M}\sum_{i\in \widehat J_s^{(m)}}\sum_{j\in \widehat J_t^{(m)}} 1
		\;\middle|\; \mathcal U \right] \quad + \frac{1}{M_{s,t}}\,
		\mathbb E \left[ \left|
		\sum_{(i,j;m) \in D_{s,t}}1 
		\;\middle|\; \mathcal U  \right| \right] \\
		&\leq 
		\mathbb E \left[
		\left|\frac{1}{\widehat M_{s,t}}-\frac{1}{M_{s,t}}\right| \widehat M_{s,t}
		\;\middle|\; \mathcal U \right]
		+\frac{1}{M_{s,t}}\,
		\mathbb E \left[ \big| D_{s,t} \big|
		\;\middle|\; \mathcal U  \right] \\
		&= 
		\frac{1}{M_{s,t}}\,\mathbb E \left[ \big| \widehat M_{s,t} -  M_{s,t}\big|
		\;\middle|\; \mathcal U  \right] +\frac{1}{M_{s,t}}\,
		\mathbb E \left[ \big| D_{s,t} \big|
		\;\middle|\; \mathcal U  \right].
	\end{align*}
	Since $|M_{s,t}-\widehat M_{s,t}|\le |D_{s,t}|$, we deduce
	\begin{equation}\label{eq:phi-hat-phi-basic}
		|\widehat\Phi_{s,t}-\Phi_{s,t}| \;\le\;\frac{2}{M_{s,t}}\mathbb E \left[ |D_{s,t}|
		\;\middle|\; \mathcal U  \right].
	\end{equation}
	Now, under the event $|\widehat U_{i,(m)}^{\mathrm{JGS}}-U_{i,(m)}|\le \eta_i^{(m)}$,
	only
	a node $i\in J_s^{(m)}$ can be misclassified   if $U_{i,(m)}$ lies within $\eta_i^{(m)}$
	of the boundary of $I_s$ with high probability. Define
	\[
	B_s^{(m)} := \Bigl\{ i \in \dblbr{ n^{(m)}}: 
	\mathrm{dist}(U_{i,(m)},\partial I_s) \le \eta_i^{(m)}\Bigr\}.
	\]
	Then every misclassified vertex of $J_s^{(m)}$ belongs to $B_s^{(m)}$.  
	If a vertex $i$ leaves $J_s^{(m)}$ due to misclassification, it can affect at most $|J_t^{(m)}|$ dyads in block $(s,t)$. By symmetry, misclassifications in $J_t^{(m)}$ contribute at most $|J_s^{(m)}|\,|B_t^{(m)}|$ dyads in block $(s,t)$. Therefore,
	\[
	|D_{s,t}|
	\;\le\;\sum_{m\in\dblbr M} \Big(|J_t^{(m)}|\,|B_s^{(m)}| + |J_s^{(m)}|\,|B_t^{(m)}|\Big).
	\]
	Plugging this into \eqref{eq:phi-hat-phi-basic}, we obtain
	\[
	|\widehat\Phi_{s,t}-\Phi_{s,t}|
	\le \frac{2}{M_{s,t}}\sum_{m\in\dblbr M}\Big(|J_t^{(m)}|\,|B_s^{(m)}|+|J_s^{(m)}|\,|B_t^{(m)}|\Big),
	\]
	and hence
	\begin{align}
		\mathbb E \left[(\widehat\Phi_{s,t}-\Phi_{s,t})^2\right]
		\;\le\; \mathbb E \left[\Bigg( \frac{2}{M_{s,t}} \sum_{m\in\dblbr M}\big(|J_t^{(m)}|\,|B_s^{(m)}|
		+|J_s^{(m)}|\,|B_t^{(m)}|\big)\Bigg)^2\right]. \label{c}
	\end{align}
	Combining the equations \eqref{eqLemmeStochError1}--\eqref{c} concludes the proof.
\end{proof}

\begin{proof}[Proof of Theorem~\ref{thm:consistency-JGS}] 
	We start from the decomposition of the $\mathrm{MISE}$
\begin{align*}
\mathbb E \left[\big\|\widehat W^{\mathrm{JGS}} - W\big\|_{L^2}^2\right] \leq  3\,\mathbb E \left[\big\|\widehat W^{\mathrm{JGS}} - W^{\mathrm{oracle}}\big\|_{L^2}^2\right] 
&\quad + 3\,\mathbb E \left[\big\|W^{\mathrm{oracle}} - W^{\mathrm{approx}}\big\|_{L^2}^2\right] + 3\,\mathbb E \left[\big\|W^{\mathrm{approx}} - W\big\|_{L^2}^2\right].
\end{align*}
By Lemma~\ref{lem:discretization_error}, Lemma~\ref{lem:histogram_true_latent},  Lemma~\ref{lem:stochastic_error}, we have
\begin{align*}
	\mathbb E \|\widehat W^{\mathrm{JGS}}-W\|_{L^2}^2 		&\leq \frac{3 C\,L^2}{k^2} + \frac{3}{4k^2}\sum_{s,t\in\dblbr k^2} \mathbb E \Big[\frac{1}{M_{s,t}} \Big] + \frac{1}{k^2}\sum_{s,t\in\dblbr k^2} \mathbb E \left[\frac{3}{4\widehat M_{s,t}}
	+ \frac{3}{4M_{s,t}}\right]  \\
	&+ \mathbb E \left[\frac{48}{M_{s,t}^2}\,
	\Bigg(\sum_{m\in\dblbr M}\big(|J_t^{(m)}|\,|B_s^{(m)}|
	+|J_s^{(m)}|\,|B_t^{(m)}|\big)\Bigg)^{\!2}\right] 
\end{align*}
	For fixed $M$, when $n_{\max} \to\infty$, we have the following asymptotic equivalences
	\[
	M_{s,t} \asymp \frac{S}{k^2}, 
	\qquad 
	\widehat M_{s,t} \asymp \frac{S}{k^2},
	\quad \text{where} \quad
	S=\sum_{m\in\dblbr M} (n^{(m)})^2,
	\]
	so that
	\[
	\frac{1}{M_{s,t}},\,\frac{1}{\widehat M_{s,t}}\lesssim \frac{k^2}{S},
	\qquad 
	\frac{1}{M_{s,t}^2}\lesssim \frac{k^4}{S^2}.
	\]
	Moreover, since the number $|J_s^{(m)}|$ of nodes from network $m$ in interval $I_s$ follows a binomial distribution  $|J_s^{(m)}|\sim\mathrm{Binomial}(n^{(m)},1/k)$ with mean $n^{(m)}/k$, Markov’s inequality yields for any $\lambda>0$,
	\[
	\mathbb P\left(|J_s^{(m)}|\ge \lambda\,\frac{n^{(m)}}{k}\right)\le \frac{1}{\lambda}.
	\]
	Hence, $|J_s^{(m)}|=O_{\mathbb P}(n^{(m)}/k)$. Moreover, for any $t$ and $s$,	\[
	|J_t^{(m)}|\,|B_s^{(m)}|+|J_s^{(m)}|\,|B_t^{(m)}|
	\;\lesssim\;\frac{n_{\max}}{k}\,\big(|B_s^{(m)}|+|B_t^{(m)}|\big).
	\]
	Combining the bounds gives
	\begin{align*}
		\mathbb E\left[\|\widehat W^{\mathrm{JGS}}-W\|_{L^2}^2\right]
		&\lesssim \frac{L^2}{k^2}
		+ \frac{k^2}{S}
		+ \frac{k^4}{S^2}\;\frac{1}{k^2}\sum_{s,t\in\dblbr k^2} 
		\frac{n_{\max}^2}{k^2}\,
		\mathbb E \left[\Bigg(\sum_{m\in\dblbr M} (|B_s^{(m)}|+|B_t^{(m)}|)\Bigg)^{2}\right].
	\end{align*}
Now,
\begin{align*}
	\mathbb E \left[\left(\sum_{m\in\dblbr M}(|B_s^{(m)}|+|B_t^{(m)}|)\right)^2\right]
	&= \sum_{m\in\dblbr M} \mathbb E \left[\left(|B_s^{(m)}|+|B_t^{(m)}|\right)^2\right]
	+ \!\!\sum_{m \neq \ell}\!
	\mathbb E\left[|B_s^{(m)}|+|B_t^{(m)}|\right]\,
	\mathbb E\left[|B_s^{(\ell)}|+|B_t^{(\ell)}|\right].
\end{align*}
Using $\mathrm{Var}(X)\le \mathbb E[X^2]$ and 
$|\mathrm{Cov}(X,Y)|\le \frac12(\mathrm{Var}(X)+\mathrm{Var}(Y))$, we obtain
\begin{align*}
\mathbb E\left[\left(|B_s^{(m)}|+|B_t^{(m)}|\right)^2\right]
	&\le 2(\mu_s^{(m)}+\mu_t^{(m)})
	+(\mu_s^{(m)}+\mu_t^{(m)})^2,
\end{align*}
where
\[
\mu_s^{(m)}=\mathbb E\left[|B_s^{(m)}|\right]
=\sum_{i\in\dblbr n^{(m)}} 
\mathbb P \big(\mathrm{dist}(U_{i,(m)},\partial I_s)\le \eta_i^{(m)}\big)
\;\le\; 2\sum_{i\in\dblbr n^{(m)}} \eta_i^{(m)}.
\]
Hence
\begin{align*}
	\mathbb E \left[\Bigg(\sum_{m\in\dblbr M}(|B_s^{(m)}|+|B_t^{(m)}|)\Bigg)^2\right]&\le 2\,\sum_{m\in\dblbr M}(\mu_s^{(m)}+\mu_t^{(m)})  +\Bigg(\sum_{m\in\dblbr M} (\mu_s^{(m)}+\mu_t^{(m)})\Bigg)^2\\
&\le 4\sum_{m\in\dblbr M}\sum_{i\in\dblbr n^{(m)}} \eta_i^{(m)}
	+\Bigg(2\sum_{m\in\dblbr M}\sum_{i\in\dblbr n^{(m)}} \eta_i^{(m)}\Bigg)^2.
\end{align*}
Putting everything together, we arrive at
\begin{align*}
	\mathbb E\left[\|\widehat W^{\mathrm{JGS}}-W\|_{L^2}^2\right]
	&\lesssim \frac{L^2}{k^2}
	+ \frac{k^2}{S}
	+ \frac{n_{\max}^2k^2}{S^2}\;
	\Bigg\{
	4\sum_{m\in\dblbr M} \sum_{i\in\dblbr n^{(m)}} \eta_i^{(m)}
	+\Bigg(2\sum_{m\in\dblbr M}\sum_{i\in\dblbr n^{(m)}} \eta_i^{(m)}\Bigg)^2
	\Bigg\}.
\end{align*}
For the  graphs of bounded size $n^{(m)}$ we have
$$
\sum_{i\in\dblbr n^{(m)}} \eta_i^{(m)} \;\le\; C,$$ 
so that
$$\sum_{m\in\dblbr M} \sum_{i\in\dblbr n^{(m)}} \eta_i^{(m)} \;\le\; M C + o(1).
$$
Similarly,
\[
\Bigg(\sum_{m\in\dblbr M} \sum_{i\in\dblbr n^{(m)}} \eta_i^{(m)}\Bigg)^2 
\;\le\; (MC+o(1))^2.
\]
Substituting into the $\mathrm{MISE}$ bound gives
\begin{align*}
	\mathbb E\left[\|\widehat W^{\mathrm{JGS}}-W\|_{L^2}^2\right]
	&\lesssim \frac{L^2}{k^2}
	+ \frac{k^2}{S}
	+ \frac{n_{\max}^2 k^2}{S^2}\,(MC+o(1))^2.
\end{align*}
Since $S=\sum_{m\in\dblbr M} (n^{(m)})^2 \gtrsim n_{\max}^2$, 
each term on the right-hand side tends to zero as $n_{\max}\to\infty$ and $k/n_{\max}\to 0$. 
Therefore,
\[
\mathrm{MISE}\!\left(\widehat W^{\mathrm{JGS}},W\right) \;\longrightarrow\; 0.
\]

\end{proof}

\subsection{Choice of the number of blocks $k$}
\label{app:k-choice}

For Lipschitz-continuous graphons, a standard bias–variance calculation for histogram estimators in the single graph setting gives
\begin{equation}
    \label{eq:mise-classical}
    \mathrm{MISE}(\widehat W)
    \;=\;
    \mathcal{O}\!\left(\frac{k^2}{n^2} + \frac{1}{k^2}\right),
\end{equation}
see, e.g., \citet[][Lemma~2.1 and the discussion after Theorem~1.2]{Gao2014}; see also \citet{Klopp2017}.
Balancing the two terms yields the heuristic choice
\[
k_{\mathrm{select}} \;\asymp\; n^{1/2},
\]
which corresponds to a risk of order $1/n$.

From a minimax perspective, \cite{Gao2014} established that in the latent (non-aligned) setting,
the optimal rate for estimating the probability matrix associated with a Lipschitz graphon under the mean-squared error is of order $\mathcal{O}((\log n)/n)$.

In the multiple-network setting, when the number of networks $M$ is fixed, the same argument applies with $n^2$ replaced
by the total number of observed dyads
\[
S \;=\; \sum_{m\in\dblbr M} (n^{(m)})^2,
\]
leading to
\begin{equation}
	\mathrm{MISE}(\widehat W)
	\;=\;
	\mathcal{O}\!\left(\frac{k^2}{S} + \frac{1}{k^2}\right)
	\qquad\text{and}\qquad 
	k_{\mathrm{select}} \;\asymp\; S^{1/4}.
\end{equation}

\paragraph{Avoiding empty blocks }
Denote $M_{s,t}$ the number of observed dyads in block $(s,t)$. To ensure all blocks are nonempty with high probability, that is $M_{s,t}>0$, we analyze the conditions on $k$.

Fix a block $(s,t)$. For off-diagonal blocks ($s \neq t$), $M_{s,t} = 0$ occurs if and only if, in every graph $m$, at least one of the classes $s$ or $t$ is empty. Define the event
\[
E_m = \left\{ |J_s^{(m)}| \geq 1\right\}\cap\left\{  |J_t^{(m)}| \geq 1 \right\}.
\]
Then $\{M_{s,t} = 0\} = \bigcap_{m\in\dblbr M} E_m^c$, and by independence across graphs,
\[
\mathbb P(M_{s,t} = 0) = \prod_{m\in\dblbr M} \mathbb P(E_m^c).
\]
Now,
\[
\mathbb P(E_m^c) \leq \mathbb P\left(|J_s^{(m)}| = 0\right) + \mathbb P\left(|J_t^{(m)}| = 0\right).
\]
Since each node is assigned to class $s$ with probability $1/k$,
\[
\mathbb P\left(|J_s^{(m)}| = 0\right) = (1 - 1/k)^{n^{(m)}} \leq e^{-n^{(m)}/k},
\]
so $\mathbb P(E_m^c) \leq 2 e^{-n^{(m)}/k}$. For diagonal blocks ($s = t$), the same bound holds since $\mathbb P(E_m^c) \leq e^{-n^{(m)}/k} \leq 2 e^{-n^{(m)}/k}$.
By combining these bounds, we obtain
\[
\mathbb P(M_{s,t} = 0) \leq 2^M \exp\!\left( - \frac{1}{k} \sum_{m\in\dblbr M} n^{(m)} \right) = \exp\!\left( -\frac{N}{k} + M \log 2 \right),
\]
where $N = \sum_{m\in\dblbr M} n^{(m)}$. A union bound over all $k^2$ blocks yields
\[
\mathbb P\left(\exists (s,t): M_{s,t} = 0\right) \leq k^2 \exp\!\left( -\frac{N}{k} + M \log 2 \right).
\]
To ensure this probability vanishes asymptotically, we require
\[
\frac{N}{k} - M \log 2 - 2 \log k \longrightarrow \infty,
\]
which is guaranteed by
\begin{equation}
	\label{eq:k-empty-blocks}
	\frac{N}{k} \;\ge\; c\,(M + \log k)
\end{equation}
for some constant $c > 0$. Since $\log k \leq \log N$, a simpler sufficient condition is
\begin{equation}
	\label{eq:k-explicit}
	k \;\le\; \frac{N}{c\,(M + \log N)}.
\end{equation}

\paragraph{Practical choice of $k$ } 
Combining the rate-optimal choice $k \asymp S^{1/4}$ with the empty-block constraint yields the practical rule
\begin{equation}\label{eq:kselect}
	k \;\asymp\; \min\!\left\{\, S^{1/4},\;\; \frac{N}{c\,(M+\log N)} \,\right\}.
\end{equation}

\textbf{Special cases }
\begin{itemize}[leftmargin=1.5em]
	\item Single graph of size $n$: $S=n^2$ gives $k\asymp \sqrt{n}$ (classical rate).
	\item $M$ graphs of comparable size $n$: $S\asymp M n^2$ gives $k\asymp M^{1/4}\sqrt{n}$, when \eqref{eq:k-explicit} holds.
	\item One dominant graph: $S\asymp n_{\max}^2$ gives $k\asymp \sqrt{n_{\max}}$.
\end{itemize}

\section{ADDITIONAL EXPERIMENTS}

\subsection{Table of true graphons}
The true graphons used in the simulation experiments are listed in Table~\ref{table_graphon_ids}.
\begin{table}[!htbp]
	\centering
	\caption{Synthetic graphons used in the simulation study (from \cite{Azizpour2025}).}
	\vspace{0.3em}
	\begin{tabular}{|c|c|}
		\hline
		\textbf{ID} & \textbf{Graphon function} \boldmath$W(u,v)$ \\
		\hline
		\multicolumn{2}{|c|}{\textbf{Monotone graphons}}\\
		\hline
		1 & \(u\,v\) \\
		2 & \(\exp \bigl(-(u^{0.7} + v^{0.7})\bigr)\) \\
		3 & \(\frac{1}{4}\,(u^2 + v^2 + \sqrt{u} + \sqrt{v})\) \\
		4 & \(\frac12(u + v)\) \\
		5 & \(\bigl[\,1 + \exp \bigl(-2(u^2 + v^2)\bigr)\bigr]^{-1}\) \\
		6 & \(\bigl[\,1 + \exp \bigl(-(\max(u,v)^2 + \min(u,v)^4)\bigr)\bigr]^{-1}\) \\
		7 & \(\exp \bigl(-\max(u,v)^{0.75}\bigr)\) \\
		8 & \(\exp \left[-\tfrac12\bigl(\min(u,v) + \sqrt{u} + \sqrt{v}\bigr)\right]\) \\
		9 & \(\log \bigl(1 + \max(u,v)\bigr)\) \\
		\hline
		\multicolumn{2}{|c|}{\textbf{Non-monotone graphons}}\\
		\hline
		10 & \(\lvert u - v\rvert\) \\
		11 & \(1 - \lvert u - v\rvert\) \\
		12 & \(0.8\, \mathrm{I}_2 \otimes \mathds{1}_{[0,1/2]^2}\) \\
		13 & \(0.8\, (\mathrm{1}-\mathrm{I}_2) \otimes \mathds{1}_{[0,1/2]^2}\) \\
		\hline
	\end{tabular}
	\label{table_graphon_ids}
\end{table}

To illustrate the type of networks generated in our simulation study, 
Figure~\ref{fig:example_graphs} displays sample graphs drawn from several 
synthetic graphons listed in Table~\ref{table_graphon_ids}. 
Each network is generated using the sampling procedure described in the main text.

\begin{figure*}[t!]
\centering
\includegraphics[width=0.24\textwidth]{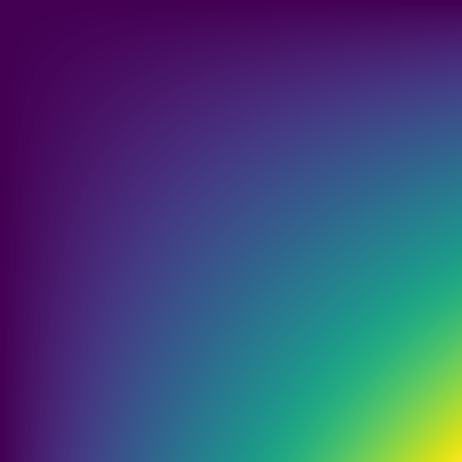}
\includegraphics[width=0.24\textwidth]{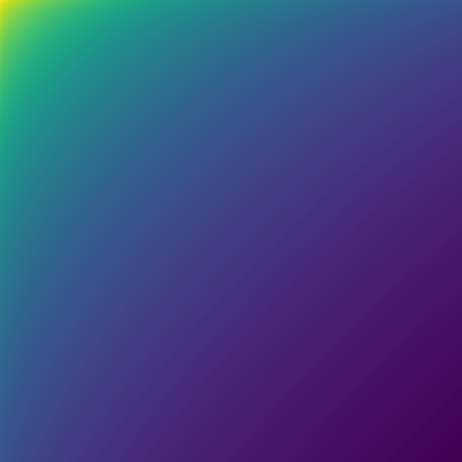}
\includegraphics[width=0.24\textwidth]{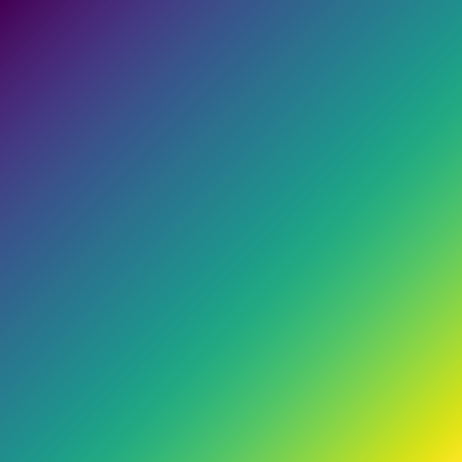}
\includegraphics[width=0.24\textwidth]{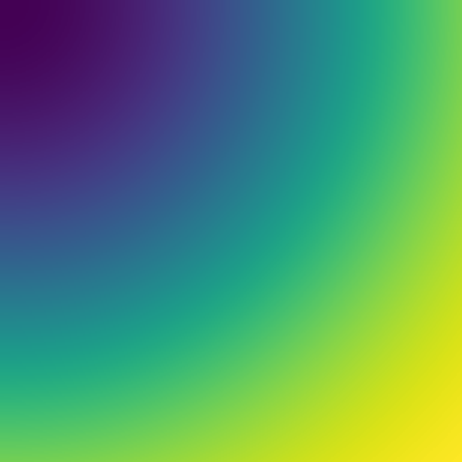}

\vspace{0.3em}

\includegraphics[width=0.24\textwidth]{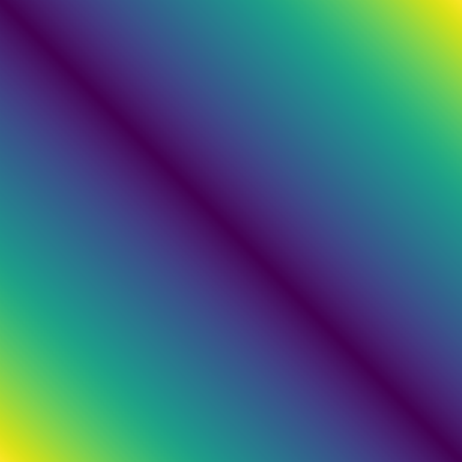}
\includegraphics[width=0.24\textwidth]{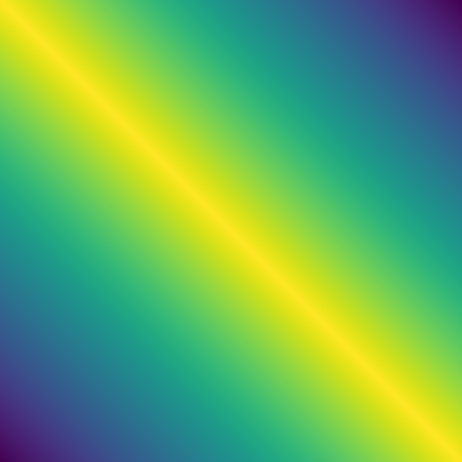}
\includegraphics[width=0.24\textwidth]{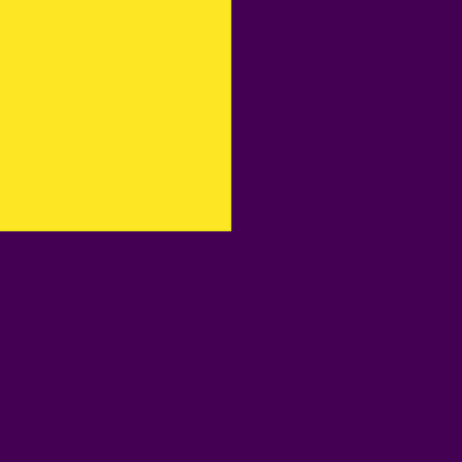}
\includegraphics[width=0.24\textwidth]{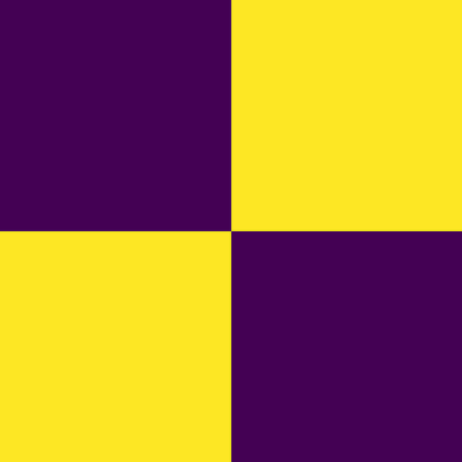}

\caption{Examples of simulated graphs generated from the synthetic graphons used in the experiments.
Each panel shows a graph sampled from a different graphon listed in Table~\ref{table_graphon_ids}.}
\label{fig:example_graphs}
\end{figure*}

\subsection{Baseline methods and hyperparameter settings}

Below we provide additional implementation details for the baseline
methods used in the simulation study. Unless otherwise stated, the
hyperparameters were chosen following the recommendations of the
original papers or of recent empirical studies. This ensures a fair and
reproducible comparison with the proposed JGS estimator.

All methods are applied to each graph separately when required by their
original formulation. In the multi-network setting considered in this
paper, the final graphon estimate is obtained either by pooling the
graphs (when the method allows it) or by averaging the estimates across
graphs. The largest graph size in the collection is denoted by
$n_{\max}$.

\begin{itemize}

\item[$\triangleright$] \textbf{SBA \citep{Chan2014}}  
The Sorting-and-Binning Algorithm estimates latent node ordering from
empirical degrees and then constructs a block histogram of the graphon.
The method depends on a threshold parameter~$\Delta$ controlling the
binning step. We adopt $\Delta = 0.1$, following the setting used in
\cite{Xu2021}.

\item[$\triangleright$] \textbf{SAS \citep{Airoldi2013}}  
The Sorting-And-Smoothing estimator first orders nodes according to
empirical degrees and then applies a histogram smoothing step. The main
hyperparameter is the histogram window size $h$. In our experiments, we
use $h = \log(n_{\max})$, as recommended in the original SAS paper and
adopted in subsequent studies such as \cite{Xu2021}. 

\item[$\triangleright$] \textbf{USVT \citep{Chatterjee2015}}  
Universal Singular Value Thresholding estimates the probability matrix
by truncating the singular value decomposition of the adjacency matrix.
The key hyperparameter is the singular value threshold $\tau$ used to
remove noisy components. We set $\tau = 0.2\,\sqrt{n_{\max}}$, following
the calibration used in \cite{Azizpour2025}. 

\item[$\triangleright$] \textbf{SGWB \citep{Xu2021}}  
The SGWB estimator relies on an optimal transport formulation to align
graphs through Gromov--Wasserstein barycenters. The method involves
several optimization parameters. We follow the settings recommended in
the original work: proximal weight $\beta = 0.005$, number of outer
iterations $L = 5$, number of Sinkhorn iterations $S = 10$, and
smoothness regularization weight $\alpha = 0.0002$. 

\item[$\triangleright$] \textbf{SIGL \citep{Azizpour2025}}  
SIGL estimates the graphon using a neural architecture combining a graph
neural network for latent position learning and an implicit neural
representation (INR) for the graphon function. We use the hyperparameter
configuration proposed in \cite{Azizpour2025}: hidden layer sizes
$\mathrm{inr\_dim\_hidden} = [20, 20]$ and
$\mathrm{gnn\_dim\_hidden} = [8, 8, 8]$, number of epochs
$\mathrm{n\_epochs} = 100$, checkpoint interval
$\mathrm{epoch\_show} = 20$, frequency parameter $\mathrm{w_0} = 10$,
learning rate $\mathrm{lr} = 0.01$, batch size
$\mathrm{batch\_size} = 1024$, and resolution $\mathrm{Res} = 1000$ used
to sample the estimated graphon.

\end{itemize}

All methods are implemented using the same graph collections and
evaluation protocol in order to ensure fair comparisons of both estimation accuracy and computational cost. All experiments were conducted on a personal workstation equipped with a 13th Gen Intel Core i5-1335U CPU running Linux.
\subsection{JGS implementation}

Algorithm~\ref{algoJGS} summarizes the implementation of the JGS
estimator. In practice, the procedure only requires computing normalized
degrees across all graphs, performing a global sorting step to obtain
latent position estimates, and constructing a block histogram of the
graphon using these estimated positions. The method naturally aggregates
information from the entire collection of networks through the pooled
ordering of nodes. Additional implementation details, including the
automatic choice of $k$ and optional post-processing steps, are provided
in the publicly available code accompanying this paper: \url{https://github.com/RolandBonifaceSogan/JGS-graphon}..

\begin{algorithm*}[!htbp]
	\caption{Joint Graph Sorting (JGS) Estimation}
	\label{algoJGS}
	
	\KwIn{
		Adjacency matrices $\{A^{(1)}, \dots, A^{(M)}\}$ of sizes $n^{(1)}, \dots, n^{(M)}$; number of blocks $k$
	}
	\KwOut{
		Estimated graphon matrix $\widehat{H}^{\mathrm{JGS}} \in [0,1]^{k \times k}$
	}
	
	\vspace{0.3em}
	\textbf{Step 1: Latent Position Estimation} \\
	\ForEach{$m \in \dblbr  M$}{
		\ForEach{$i \in \{1, \dots, n^{(m)}\}$}{
			Compute normalized degrees $\widehat d_{i,(m)}^{\mathrm{norm}} = \dfrac{1}{n^{(m)}} \sum_{j\in \dblbr{n^{(m)}}} A_{i,j}^{(m)}$ \;
		}
	}
	Concatenate all degrees into a vector of length $N = \sum_{m\in\dblbr M} n^{(m)}$ \;
	Sort this vector and denote by  $\widehat \sigma^{\mathrm{JGS}}$ the associated permutation \;
	\ForEach{node $(i,m)$}{
		Compute the latent position estimate  $\widehat{U}_{i,(m)}^{\mathrm{JGS}} = \dfrac{(\widehat \sigma^{\mathrm{JGS}})^{-1}(i,m)}{N} - \dfrac{1}{2N}$ \;
	}
	
	\vspace{0.3em}
	\textbf{Step 2: Histogram Estimation} \\
	Partition $[0,1]$ into $k$ intervals $I_1, \dots, I_k$ of length $1/k$ \;
	\ForEach{$s, t \in \dblbr k $}{
		Let $\widehat J_s{(m)} = \left\{ i : \widehat{U}_{i,(m)}^{\mathrm{JGS}} \in I_s\right \}$ and $\widehat J_t^{(m)} = \left\{ i : \widehat{U}_{i,(m)}^{\mathrm{JGS}} \in I_t \right\}$ \;
		Compute
		$
		\widehat{H}^{\mathrm{JGS}}_{s,t}=
\dfrac{\sum_{m\in\dblbr M} \sum_{i \in \widehat J_s^{(m)}} \sum_{j \in \widehat J_t^{(m)}} A^{(m)}_{i,j} }{1\vee \sum_{m\in\dblbr M} |\widehat J_s^{(m)}||\widehat J_t^{(m)}|},\quad s,t\in\dblbr k
		$
	}
	
\end{algorithm*}

\subsection{G-Mixup}

The training parameters are kept identical for the three methods to ensure a fair comparison. Optimization is performed using the Adam algorithm \citep{Kingma2014}, with an initial learning rate of 0.01, which is halved every 100 epochs. The batch size is set to 128. Datasets are split into 70\% for training, 10\% for validation, and 20\% for testing. The test accuracy reported in Table~\ref{result_real} corresponds to the accuracy obtained at the epoch achieving the best validation score, averaged over ten independent runs.

The Graph Isomorphism Network (GIN) architecture \citep{Xu2019} is used as the base classifier. Following the protocol of \cite{Han2022}, the same model architecture and training procedure are applied across all methods so that the only difference lies in the graph augmentation step. This ensures that the performance differences observed in Table~\ref{result_real} can be attributed to the quality of the graphon estimates used within the G-Mixup procedure.

In the G-Mixup framework, graphons are estimated exclusively from the training graphs and are then used to generate additional synthetic graphs via interpolation. More precisely, given two estimated class-specific graphons $W_1$ and $W_2$, a mixed graphon is constructed as
\begin{equation}
    W_\lambda = \lambda W_1 + (1-\lambda) W_2 ,
\end{equation}
where $\lambda$ is a mixing coefficient. Graphs are subsequently sampled from $W_\lambda$ to enrich the training dataset.

In our experiments, we generate an additional 20\% of training graphs using this interpolation mechanism. The mixing parameter $\lambda$ is sampled uniformly from the interval $[0.1,0.2]$, following the setup of \cite{Han2022}. Restricting $\lambda$ to this range produces synthetic graphs that remain close to the original class structures while still introducing meaningful variability in the training data.

Finally, all graphon estimators are applied independently within each class of the training data, ensuring that the generated graphs preserve class-specific structural patterns. This setup allows us to evaluate the impact of improved graphon estimation on downstream graph classification performance within the G-Mixup framework.
\subsection{Comparison of running times}

To complement the results reported in the main paper regarding the
computational efficiency of the proposed estimator, we provide here a
more detailed comparison of execution times across all considered
methods. In addition to the most accurate approaches (JGS, SIGL, and
SGWB), we also include classical estimators such as SBA, SAS, and
USVT. Although these methods generally yield lower estimation accuracy,
they are often considered computationally inexpensive and therefore
serve as useful baselines for evaluating the trade-off between
statistical accuracy and computational cost.

Table~\ref{complexity} summarizes the asymptotic computational
complexities of the different estimators in the multi-network setting,
where $N=\sum_m n^{(m)}$ denotes the total number of nodes,
$M$ the number of graphs,
$k$ the number of histogram blocks,
$S=\sum_{m\in\dblbr{M}}(n^{(m)})^2$ the total number of observed dyads,
and $L$ the number of iterations for iterative methods. As indicated by
the table, the proposed JGS estimator has one of the lowest theoretical
computational complexities.

\begin{table}[t]
\centering
\caption{Computational complexities of the considered methods in the
multi-network setting.}
\vspace{0.3em}
\begin{tabular}{|l|c|}
\hline
\textbf{Method} & \textbf{Complexity} \\
\hline
SBA  & $\mathcal{O}(M\,k\,n_{\max}\log n_{\max})$ \\
SAS  & $\mathcal{O}(M\,n_{\max}\log n_{\max} + k^2\log k^2)$ \\
USVT & $\mathcal{O}(M\,n_{\max}^3)$ \\
SGWB & $\mathcal{O}(L\,S\,M\,(S\,k + n_{\max}k^2))$ \\
SIGL & $\mathcal{O}(M\,L\,n_{\max}^2)$ \\
\textbf{JGS} & $\boldsymbol{\mathcal{O}(N\log N + S)}$ \\
\hline
\end{tabular}
\label{complexity}
\end{table}

To empirically evaluate the practical computational cost of the
different approaches, we conducted additional simulation experiments
using a fixed experimental setting with $M=50$ networks and heterogeneous
graph sizes $n^{(m)}$ around $n=100$. For each method, both the
estimation error (measured by the MISE) and the execution time were
recorded over several repetitions.

Figure~\ref{fig:time_boxplot} reports the distribution of execution
times across repetitions. As expected, iterative optimization-based
methods such as SIGL and SGWB are substantially slower than the other
approaches. In contrast, JGS remains computationally efficient and
exhibits execution times comparable to the fastest baseline methods.
This confirms that the joint sorting procedure can be implemented with
very low computational overhead even when multiple graphs are processed
simultaneously.

Figure~\ref{fig:mise_boxplot} presents the corresponding distribution of
estimation errors. These results highlight the substantial accuracy
advantage of JGS over classical methods such as SBA and SAS. While USVT
remains competitive in terms of estimation error, it does not provide a
clear advantage in terms of computational efficiency when compared with
JGS. Overall, these results illustrate that JGS achieves an attractive
balance between statistical accuracy and computational efficiency.

\begin{figure}[!htbp]
\centering
\includegraphics[width=\linewidth]{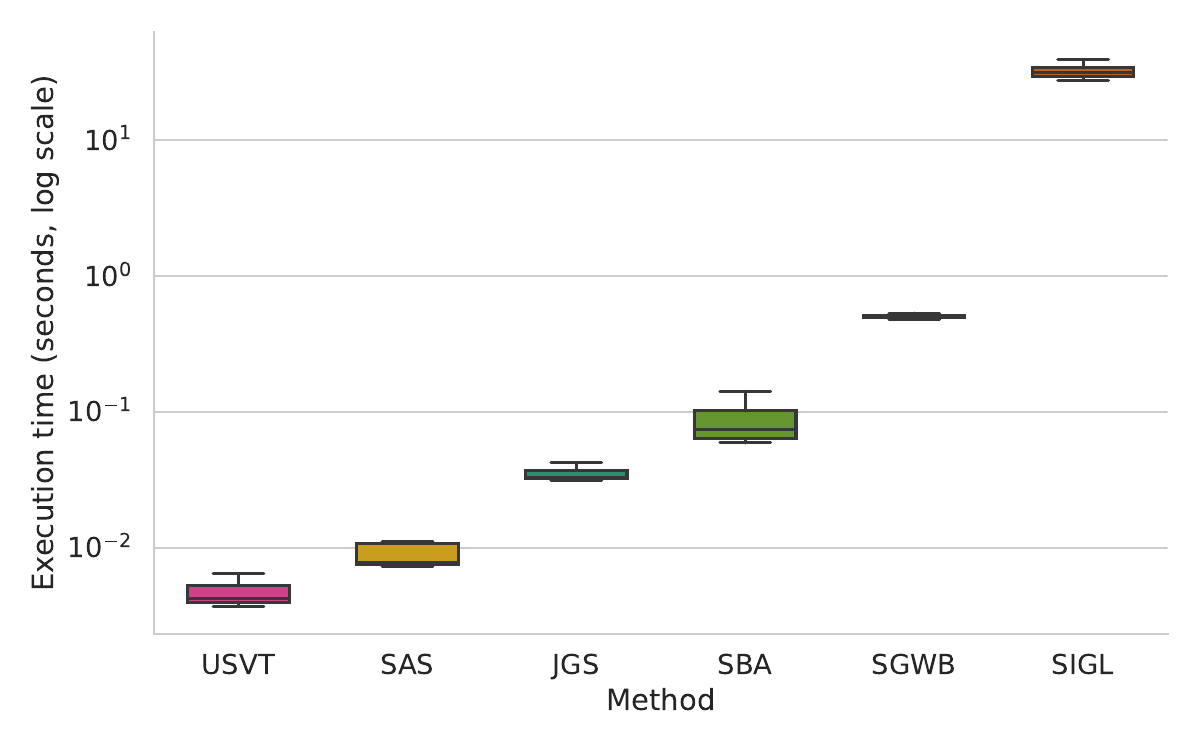}
\caption{Distribution of execution times across repetitions for the
different graphon estimation methods (log scale).}
\label{fig:time_boxplot}
\end{figure}

\begin{figure}[!htbp]
\centering
\includegraphics[width=\linewidth]{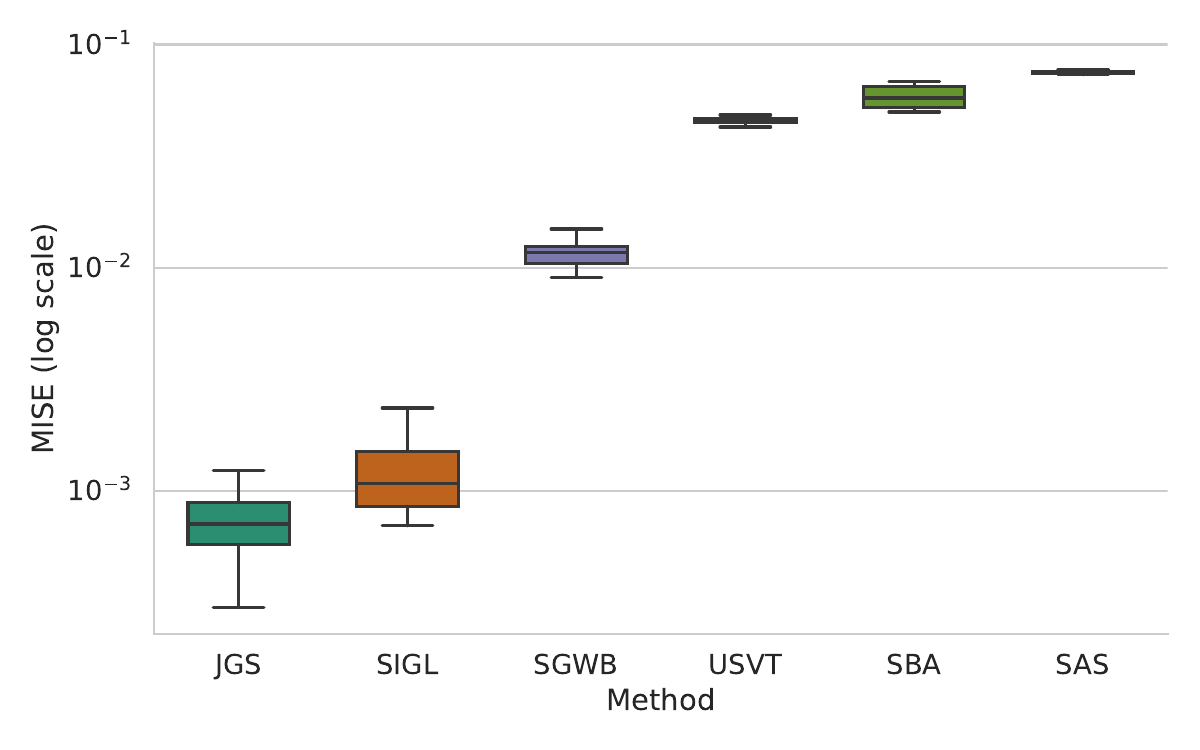}
\caption{Distribution of estimation errors (MISE) across repetitions for
the different graphon estimation methods (log scale).}
\label{fig:mise_boxplot}
\end{figure}

Together, these results confirm that JGS provides a particularly
favorable trade-off between estimation accuracy and computational cost,
making it well suited for large collections of networks.

\subsection{Sensitivity of JGS to the choice of $k$}

Figure~\ref{fig:sensitivity_k} provides a detailed view of how the
performance of JGS depends on the number of histogram blocks~$k$.
Several features can be observed.

First, for the identifiable graphons (IDs~1--9), the MISE curves
display a clear bias--variance trade-off. When $k$ is too small, the
histogram resolution is too coarse and the resulting estimator suffers
from a large approximation bias, which explains the relatively high
errors on the left-hand side of the figure. As $k$ increases, the
histogram becomes more flexible and captures the structure of the
graphon more accurately, leading to a rapid decrease of the MISE.
Beyond a certain point, however, the number of observations per block
becomes too small, which increases the estimation variance and causes
the error to rise again. This explains the characteristic U-shaped
profiles observed for most identifiable graphons.

Second, for the non-identifiable graphons (IDs~10--13), the behavior is
markedly different. The corresponding curves are much flatter and remain
at a substantially higher error level throughout the whole range of
$k$. In these cases, refining the histogram partition does not lead to
the same gain as in the identifiable setting, because the main
difficulty is no longer the approximation of the graphon by a step
function, but the lack of a canonical alignment of nodes across graphs.
As a result, the benefits of increasing $k$ are limited, and the minima
are less pronounced.

Third, the figure shows that the theoretical guideline stated in
\eqref{eq:kselect},
\[
k_{\mathrm{theoretical}}
=
\min\!\left\{\, S^{1/4},\;\; \frac{N}{2\,(M+\log N)} \,\right\},
\]
provides a practically meaningful choice of the resolution parameter.
In the present experiment, this rule gives
$k_{\mathrm{theoretical}}=24$, represented by the vertical dashed line
in Figure~\ref{fig:sensitivity_k}. For most identifiable graphons, this
value falls very close to the empirical minimizer or inside a broad
near-optimal region. This agreement is particularly encouraging because
it shows that the theoretical rule captures well the balance between
approximation error and estimation variance.

Finally, the broad flat regions around the minima indicate that the
performance of JGS is relatively robust to moderate misspecification of
$k$. In other words, choosing a value of $k$ slightly smaller or larger
than the empirical optimum often has only a minor impact on the MISE.
This is an important practical advantage, since it means that the
theoretical rule for selecting $k$ can be used directly in applications
without costly tuning procedures, while still yielding highly accurate
graphon estimates.

\begin{figure}[!htbp] 
\centering 
    \includegraphics[width=\linewidth]{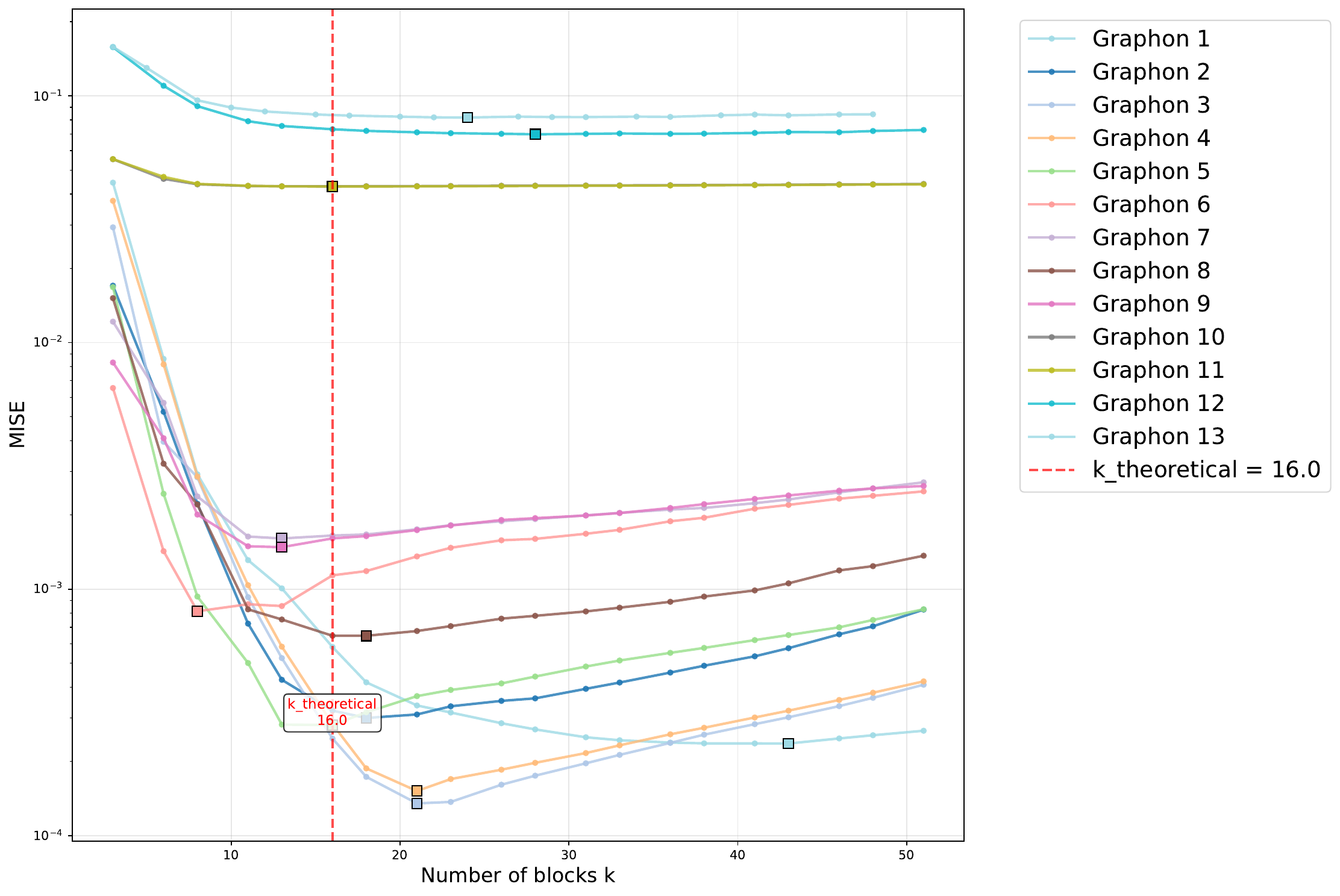}
    \caption{
Sensitivity of JGS to the choice of~$k$.
The dashed vertical line indicates the theoretical choice
$k_{\mathrm{theoretical}}=24$, while squares mark the empirical optima.
For identifiable graphons, the curves exhibit a U-shaped
behavior, whereas for non-identifiable graphons they remain flatter and
at a higher error level.
}
\label{fig:sensitivity_k}
\end{figure}

\end{document}